\documentclass[10pt,journal,epsfig]{IEEEtran}

\usepackage[dvips]{graphicx}
\usepackage{graphicx}
\usepackage{amssymb}
\usepackage{cite}
\usepackage{subfigure}
\usepackage{amsmath}

\usepackage{subeqnarray}
\usepackage{cases}
\usepackage[centerlast]{caption2}
\usepackage{color}

\begin{document}

\IEEEoverridecommandlockouts
\title{
Linear Precoder Design for a MIMO Gaussian Wiretap Channel with Full-Duplex Source and Destination Nodes
}
\author{Lingxiang Li, Zhi Chen, ~\IEEEmembership{Member,~IEEE}, \\ Athina P. Petropulu, ~\IEEEmembership{Fellow,~IEEE}, and Jun Fang, ~\IEEEmembership{Member,~IEEE}
\thanks{Lingxiang Li and Zhi Chen and Jun Fang are with the National Key Laboratory of Science and Technology on Communications,
University of Electronic Science and Technology of China, Chengdu 611731, China (e-mails:lingxiang.li@rutgers.edu; \{chenzhi, JunFang\}@uestc.edu.cn).
The work was performed when L. Li was a visiting student at Rutgers University.}
\thanks{A. P. Petropulu is with the Department of Electrical and Computer Engineering, Rutgers--The State University of New Jersey, New Brunswick, NJ 08854 USA (e-mail: athinap@rci.rutgers.edu).}
}


\maketitle

\begin{abstract}
We consider linear precoder design for a multiple-input multiple-output (MIMO) Gaussian wiretap channel, which
comprises two legitimate nodes, i.e., \emph{Alice} and \emph{Bob}, operating in Full-Duplex (FD) mode and exchanging
confidential messages in the presence of a passive eavesdropper.
Using the sum secrecy degrees of freedoms (sum S.D.o.F.) as
reliability measure, we formulate an optimization problem with respect to the precoding matrices.
In order to solve this problem, we first propose
a cooperative secrecy transmission scheme, and prove that its
feasible set is sufficient to achieve the maximum sum S.D.o.F.. Based on that feasible set,
we then determine the maximum achievable sum S.D.o.F. in closed form, and provide a method for constructing the
precoding matrix pair which achieves the maximum sum S.D.o.F..
Results show that, 
the FD based network provides an attractive secrecy transmission rate performance.

\end{abstract}


\section{Introduction}

Full-duplex (FD) has attracted intensive attention in the past few years since it has the potential to double the spectral efficiency.
Due to the challenges in suppressing self-interference,
which is inherent to FD, wireless communication
systems have largely avoided FD up to recently. However,
as short-range systems with low-power transmitters such as
small-cell systems and WiFi are becoming dominant, there has been
renewed interest in FD, since self-interference in such systems
is more manageable \cite{Sabharwal14,Lingyang15}.

Recently, FD nodes were used in the context of physical layer secrecy.
One line of research considers an FD receiver (\emph{Bob}) \cite{Wei12,Gaojie15,Fengchao16}, who
transmits jamming signals, which overlap in time and frequency with the source's (\emph{Alice}) signal.
Introducing jamming signals in the from of artificial noise
is an effective way to improve secrecy \cite{Negi05,qiang13,Tang08,LunDong10,Zheng11,Han11,Swindlehurst11,zheng151,Lingxiang16}, since
jamming signals can be designed to degrade the eavesdropper's (\emph{Eve}) channel without hurting the legitimate channel.
Typically, the jamming signals are transmitted by the transmitter \cite{Negi05,qiang13}, or
from external helpers \cite{Tang08,LunDong10,Zheng11,Han11,Swindlehurst11,zheng151,Lingxiang16}. 
Jamming by the transmitter does not have to rely on external helpers,
who may not be trustworthy, or maybe be moving and thus hard to keep track of.
Multi-antenna techniques have been used to further boost the
potential benefits of using an FD \emph{Bob}. Specifically,
\cite{Gan13,Yongkai14,Yongkai14spl} proposed algorithms for maximizing the secrecy rate
over the covariance matrix of jamming signals,
while \cite{LLX16} studied the maximum secrecy degrees
of freedom (S.D.o.F.), and uncovered its connection to the number of antennas at each node.


Using an FD transmitter as well as an FD receiver has the potential to substantially
improve the achievable secrecy rate.
Further, when both transmitter and receiver have multiple antennas, the transmission of a
given node during the reception of information can be designed to act as jamming signal and degrade \emph{Eve}'s channel.
This is considered in \cite{Yaping15,Renhai16,Cepheli14}, where the bidirectional communication
creates co-channel interference (CCI), which can act as an alternative to jamming
for the purpose of degrading \emph{Eve}'s channel \cite{Kalantari15,Lv15,Lingxiang162,Xie15,Xie142,Koyluoglu112,Tung15}.
In particular, the works \cite{Yaping15,Renhai16} assumed that each terminal receives with a single antenna, and proposed
algorithms to find a beamforming design that maximizes the achievable secrecy rate;
the work of \cite{Cepheli14} assumed that each terminal transmits and receives with multiple antennas,
and proposed algorithms to find the beamforming design that minimizes the transmit power subject to
certain quality of service (QoS) requirements.

In this paper, we consider the general multi-input and multi-output (MIMO) Gaussian wiretap channel as in \cite{Cepheli14}, i.e.,
a network comprising two FD legitimate nodes \emph{Alice} and \emph{Bob},
and a passive eavesdropper \emph{Eve}. Unlike \cite{Cepheli14}, which assumes that each transmitter
sends a single signal stream and tries to minimize the transmit power subject to
certain QoS requirements, we consider the multiple signal streams case and
our goal is to maximize the achievable sum secrecy rate via the proper design of precoding matrices.
Due to the self-interference, the achievable secrecy rate of each link is a nonlinear fractional function of the precoding matrices.
This makes the sum secrecy rate maximization problem a difficult problem to solve.
Instead, we consider the sum S.D.o.F. as a surrogate, 
i.e., the rate at which the achievable sum secrecy rate scales with ${\rm log}(P)$ in the high signal-to-noise ratio (SNR) regime.

Our main contributions are summarized below.

We propose a design for the precoding matrices of \emph{Alice} and \emph{Bob}, with which the maximum
sum S.D.o.F. is achieved.
This is achieved in the following steps. First, we propose a cooperative secrecy transmission scheme, in which the
message signals from \emph{Alice} and \emph{Bob} are aligned along the same received subspace of \emph{Eve}.
We then prove that the maximum sum S.D.o.F. can be achieved by the precoding matrices that include the largest possible number of
precoding vectors  produced by the proposed scheme, which are linearly independent and interference free.
Subsequently, we divide the candidate precoding vector pairs 
into several subsets, based on their potential to
achieve a greater sum S.D.o.F.. For each subset, we provide the number of linearly independent pairs and their mathematical description.
Finally, we give an algorithm (see Table II)
for selecting the precoding pairs from the various subsets, so that the sum S.D.o.F. is maximized.
We also determine the maximum achievable sum S.D.o.F. as a function of the number of antennas
(see equations (\ref{eq40})-(\ref{eq49})). 
Our analytical results show exactly how the sum secrecy rate depends on the number of antennas at \emph{Alice}, \emph{Bob} and \emph{Eve}.


In \cite{Lingxiang162}, we determined the maximum achievable S.D.o.F. region of a two-user wiretap channel with a source destination pair exchanging confidential messages, another pair exchanging public messages, and a passive eavesdropper who is interested in the communications of the former pair.
In this paper, while the methodology is similar to that of \cite{Lingxiang162}, the problem is different because, unlike \cite{Lingxiang162},
in this paper, \emph{Eve} has interest in both source signals.
This makes the S.D.o.F. region maximization problem significantly more difficult. In particular, the problem becomes equivalent to
two nonlinear fractional problems (each corresponding to the secrecy rate of a wiretap channel).
Therefore, the S.D.o.F. region maximization problem is more complicated and the result cannot be obtained through a straightforward extension of
\cite{Lingxiang162}.




The rest of this paper is organized as follows. In
Section II, we describe the system model
and formulate the sum S.D.o.F. maximization problem. In Section III, we propose a
secrecy cooperative transmission scheme, and prove that its feasible set is sufficient
to achieve the maximum sum S.D.o.F.. In Section IV, we divide the feasible set of precoding vectors 
into several subsets.
For each subset, we derive the formulas of the precoding vectors and
determine the number of linearly independent candidate precoding vectors.
In Section V, we give the maximum achievable sum S.D.o.F. as a function of the number of antennas,
and we also provide a method for constructing the precoding matrix pair which achieves the maximum sum S.D.o.F..
Numerical results are given in Section VI and conclusions are drawn in Section VII.

\textit{Notation:}
$x\sim\mathcal{CN}(0,\Sigma)$ means $x$ is a random variable following a complex circular Gaussian
distribution with mean zero and covariance $\Sigma$; $\lfloor a\rfloor$ denotes the largest integer which is less or equal to $a$;
$(a)^+ \triangleq \max \{a,0\}$;
$\min^+\{a, b\} \triangleq (\min\{a, b\})^+$.
We use lower case bold to denote vectors;
$\mathbb{C}^{N \times M}$ indicates a ${N \times M}$ complex matrix set;
${\bf{A}}^T$, ${\bf{A}}^H$, $\rm{tr}\{\bf{A}\}$,
$\rm{rank}\{\bf{A}\}$, and $|{\bf{A}}|$ stand for the transpose, hermitian transpose, trace,
rank and determinant of the matrix $\bf{A}$, respectively; ${\bf A}(:,j) $
indicates the $j$-th column of $\bf A$; 
${\rm {span}}({\bf A})$ and ${\rm {span}}({\bf A})^\perp$ are the subspace spanned by
the columns of $\bf A$ and its orthogonal complement, respectively;
${\rm {dim}}\{{\rm {span}}(\bf A)\}$ represents
the number of dimension of the subspace spanned by the columns of $\bf A$;
${\rm {null}}({\bf A})$ denotes the null space of ${\bf A}$;
${\bf \Gamma}({\bf A})$ denotes the orthogonal basis of ${\rm{null}}({\bf A})$;
${\bf A}^\perp$ denotes the orthogonal basis of ${\rm{null}}({\bf A}^H)$; ${\rm {span}}({\bf A})\cap{\rm {span}}({\bf B})$
denotes the intersection of the subspaces ${\rm {span}}({\bf A})$ and ${\rm {span}}({\bf B})$;
{${\rm {span}}({\bf A})\setminus{\rm {span}}({\bf B})\triangleq\{{{\bf{x}}|{\bf{x}}\in{\rm {span}}({\bf A}),
{\bf{x}} \notin {\rm {span}}({\bf B})}\}$}.
${\bf I}$ represents an identity matrix with appropriate size.
We denote by ${\mathcal I} \succ {\mathcal J}$ to indicate that we pick
precoding vector pairs from ${\mathcal I}$ prior to ${\mathcal J}$;
we denote by ${\mathcal I} = {\mathcal J}$ to indicate that
we can select precoding vector pairs from
${\mathcal I}$ and ${\mathcal J}$ without any specific constraints on rankings.

\newtheorem{proposition}{Proposition}
\newtheorem{theorem}{Theorem}
\newtheorem{corollary}{Corollary}
\newtheorem{lemma}{Lemma}

\begin{figure}[!t]
\centering
\includegraphics[width=3in]{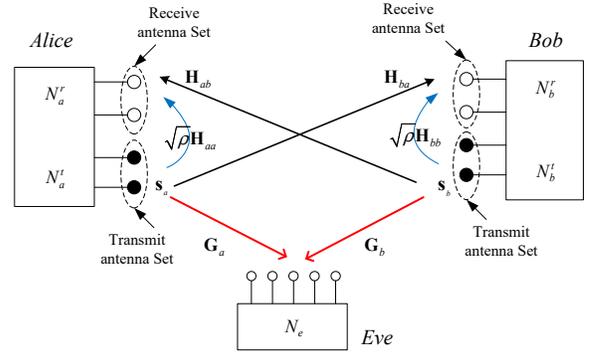}
\DeclareGraphicsExtensions. \caption{A MIMO FD bidirectional wiretap channel.}
\vspace* {-12pt}
\end{figure}
\section{System Model and Problem Statement}
We consider a MIMO Gaussian wiretap channel (see Fig. 1), which consists of two
legitimate transceivers, \emph{Alice} and \emph{Bob}, who want to exchange information,
and an external passive eavesdropper, \emph{Eve}, who has interest in the message signals sent by both \emph{Alice} and \emph{Bob}.
\emph{Alice} and \emph{Bob} are equipped with $N_a$ and
$N_b$ antennas, respectively. \emph{Eve} is equipped with $N_e$ antennas.
To enable simultaneous information exchange, both \emph{Alice} and \emph{Bob} operate in FD mode, i.e.,
each of them is equipped with two groups of RF chains and corresponding
antennas, one for transmitting and one for receiving.
Specifically, \emph{Alice} allocates $N_a^r$ antennas to receive
and the remaining $N_a^t=N_a-N_a^r$ antennas to transmit.
\emph{Bob} allocates $N_b^r$ antennas to receive
and the remaining $N_b^t=N_b-N_b^r$ antennas to transmit.
We denote by ${\bf s}_a\sim \mathcal{CN}(\bf{0},\bf{I})$ and
${\bf s}_b\sim \mathcal{CN}(\bf{0},\bf{I})$ the message signals sent by \emph{Alice} and \emph{Bob}, respectively.
Both signals are transmitted simultaneously and over the same frequency spectrum.
Such transmission leads to self-interference, e.g., \emph{Bob} will
also see the signals sent from its own transmit antennas, i.e., ${\bf s}_b$, whose intended receiver is \emph{Alice}.
There are various self-interference cancelation techniques such as antenna isolation, analog-circuit-domain based methods and
digital-domain based methods, however, today's state-of-the-art cannot achieve full self-interference cancelation \cite{Sabharwal14}.
To describe the effect of residual self-interference
we employe the loop interference model as in \cite{Gan13}, with parameter
$\rho=0$ corresponding to the no residual self-interference case and $0<\rho \le 1$ corresponding to different residual self-interference levels.
The signal received at \emph{Alice} and \emph{Bob} can thus be expressed respectively as
\begin{subequations}
\begin{align}
&{{\bf y}_a} =\sqrt{\rho} {\bf{H}}_{aa}{\bf V}_a{\bf s}_a + {{\bf{H}}_{ab}}{{\bf V}_b{\bf s}_b} + {{\bf{n}}_a }, \label{eq2}\\
&{{\bf y}_b} = {\bf{H}}_{ba}{\bf V}_a{\bf s}_a + \sqrt{\rho}{{\bf{H}}_{bb}}{{\bf V}_b{\bf s}_b} + {{\bf{n}}_b }. \label{eq1}
\end{align}
\end{subequations}
The signal received at \emph{Eve} can be expressed as
\begin{align}
{{\bf{y}}_e} = {{\bf G}_a}{\bf V}_a{\bf s}_a + {{\bf G}_b}{{\bf V}_b{\bf s}_b} + {{\bf{n}}_e}\textrm{.} \label{eq0}
\end{align}
Here, ${\bf V}_a$ and ${\bf V}_b$ are the precoding matrices at \emph{Alice} and \emph{Bob}, respectively;
${{\bf{n}}_a} \sim \mathcal{CN}(\bf{0},\bf{I}) $, ${{\bf{n}}_b} \sim \mathcal{CN}(\bf{0},\bf{I}) $
and ${{\bf{n}}_e} \sim \mathcal{CN}(\bf{0},\bf{I})$ are independent AWGN vectors, and
represent the measurement noise at \emph{Alice}, \emph{Bob} and \emph{Eve}, respectively;
${\bf{H}}_{ba}\in\mathbb{C}^{N_b^r \times N_a^t}$ denotes the channel matrix from
\emph{Alice} to \emph{Bob};
${\bf{H}}_{ab}\in\mathbb{C}^{N_a^r \times N_b^t}$ denotes the channel matrix from \emph{Bob} to \emph{Alice};
${\bf{H}}_{aa}\in\mathbb{C}^{N_a^r \times N_a^t}$ and ${\bf{H}}_{bb}\in\mathbb{C}^{N_b^r \times N_b^t}$ are self-interfering matrices;
${\bf G}_a\in\mathbb{C}^{N_e \times N_a^t}$ and ${\bf G}_b\in\mathbb{C}^{N_e \times N_b^t}$ denote the channel matrix from
\emph{Alice} and \emph{Bob} to \emph{Eve}, respectively.

In this paper, we make the following assumptions:
\begin{enumerate}
\item The messages ${\bf s}_a$ and ${\bf s}_b$ are independent of each other,
and independent of the noise vectors.
The receivers do not have the capability of multiple-user decoding and they will treat the interference simply as noise.
\item All the channels are flat fading and independent of each other; the corresponding channel matrices are full rank. Global channel state
information (CSI) is available at the legitimate nodes, including the CSI of \emph{Eve}. This is possible in situations in which
\emph{Eve} is a passive network user and its whereabouts and behavior can be monitored.
\end{enumerate}

For a given precoding matrix pair $({\bf V}_a, {\bf V}_b)$,
the maximum achievable rate at the legitimate receiver
and the eavesdropper can be respectively expressed as \cite{OggierBabak11}
\begin{subequations}
\begin{align}
& R_a = {\rm {log}}|{\bf I}+({\bf I}+\rho{\bf{H}}_{bb}{\bf Q}_b{\bf{H}}_{bb}^H)^{-1}{\bf{H}}_{ba}{\bf Q}_a{\bf{H}}_{ba}^H|, \label{eq3a}\\
& R_b = {\rm {log}}|{\bf I}+({\bf I}+\rho{\bf{H}}_{aa}{\bf Q}_a{\bf{H}}_{aa}^H)^{-1}{\bf{H}}_{ab}{\bf Q}_b{\bf{H}}_{ab}^H|, \label{eq3b}\\
& R_a^e = {\rm {log}}|{\bf I}+({\bf I}+{\bf G}_b{\bf Q}_b{\bf G}_b^H)^{-1}{\bf G}_a{\bf Q}_a{\bf G}_a^H|,\label{eq3c}\\
& R_b^e = {\rm {log}}|{\bf I}+({\bf I}+{\bf G}_a{\bf Q}_a{\bf G}_a^H)^{-1}{\bf G}_b{\bf Q}_b{\bf G}_b^H|, \label{eq3d}
\end{align}
\end{subequations}
where ${\bf{Q}}_a\triangleq{\bf V}_a{\bf V}_a^H$ and ${\bf Q}_b\triangleq{\bf V}_b{\bf V}_b^H$ denote the transmit
covariance matrices of \emph{Alice} and \emph{Bob}, respectively.


Correspondingly, the achievable S.D.o.F. is \cite{Liang09}
\begin{align}
d_s^i\triangleq \mathop{\lim }\limits_{ P \to \infty }{R_s^i}/{{\rm log} \ P},\ i =a, b, \label{eq6}
\end{align}
where $P$ denotes the transmit power budget, $R_s^i$ the secrecy rate which equals
\begin{align}
R_s^i\triangleq (R_i-R_e^i)^+. \label{eqSR}
\end{align}

Let the maximum achievable sum S.D.o.F. over the precoding matrices be
\begin{align}
d_s^{\rm sum}  \triangleq &\mathop {\max }\nolimits_{({\bf V}_a, {\bf V}_b) \in {\mathcal I} }(d_s^a+d_s^b), \label{eq5}
\end{align}
with ${\mathcal I}\triangleq \{({\bf V}_a, {\bf V}_b)|{\rm{tr}} \{ {\bf V}_a{\bf V}_a^H\}=P, {\rm{tr}} \{ {\bf V}_b{\bf V}_b^H\}= P\}$.
In this paper,
we aim to determine $d_s^{\rm sum}$ as a function of the number of antennas,
and thus provide some insight into the potential benefits that can be brought by FD operations.
To that objective, in the following sections, we will first introduce a cooperative
transmission scheme which can achieve the maximum sum S.D.o.F..
Subsequently, by studying the cooperative
transmission scheme, we will determine $d_s^{\rm sum}$ in closed form
and also provide the precoding matrix pair which achieves the sum S.D.o.F. of $d_s^{\rm sum}$.

\section{Cooperative Secrecy Transmission Scheme}
Before proceeding, please refer to Appendix A for some mathematical background on generalized
singular value decomposition (GSVD), which provides a mathematical basis for the text
to follow.

\begin{lemma}
For any given precoding matrices $({\bf V}_a,{\bf V}_b)\in {\mathcal I}$
the achieved S.D.o.F. can be re-expressed as follows:
\begin{subequations}
\begin{align}
&d_s^a({\bf V}_a,{\bf V}_b)= m_1({\bf V}_a,{\bf V}_b)-n_1({\bf V}_a,{\bf V}_b), \label{eqa2a}\\
&d_s^b({\bf V}_a,{\bf V}_b)= m_2({\bf V}_a,{\bf V}_b)-n_2({\bf V}_a,{\bf V}_b), \label{eqa2b}
\end{align}
\end{subequations}
in which
\begin{align}
&m_1({\bf V}_a,{\bf V}_b)\triangleq  {\rm{dim}}\{{\rm{span}}({\bf{H}}_{ba}{{\bf V}_a})\setminus{\rm{span}}({\bf{H}}_{bb}{\bf V}_b)\},  \nonumber \\
&n_1({\bf V}_a,{\bf V}_b)\triangleq {\rm{dim}}\{{\rm{span}}({\bf G}_a{{\bf V}_a})\setminus{\rm{span}}({\bf G}_b{\bf V}_b)\}, \nonumber \\
&m_2({\bf V}_a,{\bf V}_b)\triangleq {\rm{dim}}\{{\rm{span}}({\bf{H}}_{ab}{\bf V}_b)\setminus{\rm{span}}({\bf{H}}_{aa}{{\bf V}_a})\}, \nonumber \\
&n_2({\bf V}_a,{\bf V}_b)\triangleq {\rm{dim}}\{{\rm{span}}({\bf G}_b{\bf V}_b)\setminus{\rm{span}}({\bf G}_a{{\bf V}_a})\}. \nonumber
\end{align}
\end{lemma}
\begin{proof}
The proof is omitted since it's similar to that of the equation (13a) in \cite{Lingxiang162}.
\end{proof}

With \emph{Lemma 1}, one can see that
the achievable S.D.o.F. each legitimate channel can offer, is equal to the dimension difference of the interference free
subspaces which the intended destination and \emph{Eve} can respectively see.
Motivated by this observation,
we propose a cooperative secrecy transmission scheme in which the message signals from \emph{Alice}
and \emph{Bob} are aligned along the same received subspace of \emph{Eve}, i.e.,
the set of precoding matrix pairs that meet the requirements of
the proposed scheme can be expressed as follows:
\begin{align}
\bar{\mathcal I} =\{({\bf V}_a, {\bf V}_b)|{\rm{span}}({\bf G}_a{{\bf V}_a}) = {\rm{span}}({\bf G}_b{\bf V}_b),
({\bf V}_a, {\bf V}_b) \in {\mathcal I}\}. \nonumber
\end{align}
In this way, \emph{Eve}
can only see a distorted version of the message signal, and thus both $R_e^a$ and $R_e^b$
converge to a constant as $P$ approaches infinity.

In order to solve the sum S.D.o.F. maximization problem, as in \cite{Lingxiang162},
we propose to align the signals from \emph{Alice} and \emph{Bob} along the same received subspace of \emph{Eve}.
However, due to the fact that \emph{Eve} has interest in both source signals, it
does not require the total signal streams the legitimate receiver can see to be no greater than the total number of receive antennas,
and thus we get a new transmission scheme.
Based on their potential to achieve a greater sum S.D.o.F.,
in the next section we will reclassify the candidate precoding vector pairs into
eight subsets, determine the number of linearly independent candidate precoding vector
pairs in each subset, and give their rankings in the construction of the precoding matrix pair.
It turns out that the proposed scheme is sufficient to achieve the maximum sum S.D.o.F.. Details are
given by the following proposition.




\begin{proposition}
Let
\begin{align}
\bar{d}_s^{\rm sum} \triangleq &\mathop {\max }\nolimits_{ ({\bf V}_a, {\bf V}_b) \in \bar {\mathcal I} }  (d_s^a+d_s^b). \label{eq7}
\end{align}
Then $d_s^{\rm sum} =\bar{d}_s^{\rm sum} $.
\end{proposition}
\begin{proof}
See Appendix B.
\end{proof}

By \emph{Proposition 1}, we preclude a large number of precoding matrices, which have no
contribution to the maximum achievable value of the sum S.D.o.F,
and thus reduce the number of precoding matrices we need to investigate.
In the sequel, we give \emph{Corollary 1}, by which we further reduce the candidate precoding matrices.

\begin{corollary}
Let
\begin{align}
\hat{d}_s^{\rm sum}  \triangleq &\mathop {\max }\nolimits_{({\bf V}_a, {\bf V}_b) \in \hat {\mathcal I}}  (d_s^a+d_s^b), \label{eq8}
\end{align}
with $\hat{\mathcal I} =\{({\bf V}_a, {\bf V}_b)|{\bf G}_a{{\bf V}_a} ={\bf G}_b{\bf V}_b,
({\bf V}_a, {\bf V}_b) \in \bar {\mathcal I}\}$.
Then, $\bar{d}_s^{\rm sum}  =\hat{d}_s^{\rm sum} $.
\end{corollary}
\begin{proof}
See Appendix C.
\end{proof}

\section{Feasible Set of Precoding Vector Pairs of the Proposed Scheme}
The combination of \emph{Proposition 1} and \emph{Corollary 1} indicates
that for the purpose of obtaining the maximum sum S.D.o.F., we only need to investigate the maximum
achievable sum S.D.o.F. over the set of precoding matrix pairs $\hat{\mathcal I}$.
Let $({\bf v}_a, {\bf v}_b)$ be the precoding vectors comprising $({\bf V}_a, {\bf V}_b) \in \hat{\mathcal I}$ (see
the definition of $\hat{\mathcal I}$ under equation (\ref{eq8})).
In this section, we construct $({\bf V}_a, {\bf V}_b)$ one vector pair $({\bf v}_a, {\bf v}_b)$ at a time.

Some observations are in order. First, obviously we are interested in linearly independent precoding vectors.
Second, one can see that when the message signal sent by one source falls into the
null space of the eavesdropping channel, 
the interference from the other source cannot degrade
any further the eavesdropping channel because \emph{Eve} already receives nothing;
in those cases we may take the precoding vector at the other source to be zero. 
Third, for any precoding matrix pairs $({\bf V}_a, {\bf V}_b) \in \hat{\mathcal I}$, since
${\bf G}_a{{\bf V}_a} ={\bf G}_b{\bf V}_b$, it holds that
$n_1({\bf V}_a,{\bf V}_b)=n_2({\bf V}_a,{\bf V}_b)=0$, 
which combined with \emph{Lemma 1}, indicates
\begin{subequations}
\begin{align}
&d_s^a={\rm {dim}}\{{\rm {span}}({\bf H}_{ba}{\bf V}_a)\setminus{\rm {span}}({\bf H}_{bb}{\bf V}_b)\}, \label{eqi1a} \\
&d_s^b={\rm {dim}}\{{\rm {span}}({\bf H}_{ab}{\bf V}_b)\setminus{\rm {span}}({\bf H}_{aa}{\bf V}_a)\}. \label{eqi1b}
\end{align}
\end{subequations}
Thus, the sum of $d_s^a$ and $d_s^b$ increases as we include more linearly independent interference free
precoding vector pairs in $({\bf V}_a, {\bf V}_b)$.
Fourth, since all the channel matrices are assumed to be full rank, and via (\ref{eqi1a}), (\ref{eqi1b}) one can
see that
\begin{subequations}
\begin{align}
&d_s^a=\min\{(N_b^r-{\rm {rank}}\{{\bf H}_{bb}{\bf V}_b\})^+, {\rm {rank}}\{{\bf H}_{ba}{\bf V}_a\}\}, \label{eqc1a} \\
&d_s^b=\min\{(N_a^r-{\rm {rank}}\{{\bf H}_{aa}{\bf V}_a\})^+, {\rm {rank}}\{{\bf H}_{ab}{\bf V}_b\}\}, \label{eqc1b}
\end{align}
\end{subequations}
where the first term in the min operator denotes the dimension of the interference free subspace
\emph{Bob} and \emph{Alice} can see, respectively; the second term in the min operator 
represents the number of message signal streams \emph{Bob} and \emph{Alice} can see, respectively.
Motivated by these observations, we next divide the set of precoding vector pairs
into eight subsets, namely, $Sub_{\rm 11}$,..., $Sub_{\rm 14}$, $Sub_{\rm 21}$,..., $Sub_{\rm 24}$.

$Sub_{  {1i}}$: Either ${\bf v}_a$ or ${\bf v}_b$ falls
into the null space of the eavesdropping channel. For the pairs in $Sub_{{11}} \cup Sub_{{12}}$
it holds that ${\bf v}_b = {\bf 0}$; for the pairs in $Sub_{{13}} \cup  Sub_{{14}}$ it holds that
${\bf v}_a = {\bf 0}$. \emph{Alice} is self-interference free
and suffers from self-interference for the pairs in $Sub_{  {11}} \cup  Sub_{{12}}$, respectively.
\emph{Bob} is self-interference free and suffers from self-interference for the pairs in
$Sub_{  {13}}\cup Sub_{  {14}}$, respectively.

$Sub_{  {2i}}$: Both ${\bf v}_a$ and ${\bf v}_b$ 
do not lie within the null space of the eavesdropping channel. For the pairs in $Sub_{{21}}$,
both \emph{Alice} and \emph{Bob} are self-interference free; for the pairs in $Sub_{{22}}$, \emph{Bob} is self-interference free, but \emph{Alice} suffers from self-interference; for the pairs in $Sub_{{23}}$, \emph{Alice} is self-interference free, but \emph{Bob} suffers from self-interference; for the pairs in $Sub_{{24}}$, both \emph{Alice} and \emph{Bob} suffer from self-interference.


In the sequel, we will first derive the formula for ${\bf v}_a$ and ${\bf v}_b$ in each subset.
As it will become clear, the formula for $({\bf v}_a, {\bf v}_b)$ in different subsets may have some common basis vectors;
in those cases, and since we are interested in linearly independent ${\bf v}_a$'s and ${\bf v}_b$'s,
the common basis vectors will only be attributed to the subset with the highest priority,
i.e., its precoding vector pairs have the potential to achieve a greater sum of $d_s^a$ and $ d_s^b$.
Based on these observations, we then determine the number of linearly independent candidate precoding vector pairs
in $Sub_{11}$,..., $Sub_{14}$, $Sub_{21}$,..., $Sub_{24}$, i.e., $d_{ 11}$,...,
$d_{ 14}$, $d_{ 21}$,..., $d_{ 24}$, respectively.

\subsection{The formulas for ${\bf v}_a$ and ${\bf v}_b$ in each subset}

1) $Sub_{11}$: The precoding vectors in $Sub_{11}$ should satisfy
\begin{subequations}
\begin{align}
&{\bf G}_a{\bf v}_a={\bf 0},  \label{eq9a}\\
&{\bf{H}}_{aa}{\bf v}_a={\bf 0}. \label{eq9b}
\end{align}
\end{subequations}

By definition, it holds that ${\bf{G}}_{a}{\bf v}_a={\bf{G}}_{b}{\bf v}_b={ \bf 0}$.
In this subset we will only consider ${\bf v}_b = 0$,
since even if ${\bf v}_b \ne 0$ the interference from \emph{Bob} cannot degrade any further the eavesdropping channel.

Substituting ${\bf v}_a={\bf \Gamma}({\bf G}_{a}){\bf x}$ into (\ref{eq9b}), with $\bf x$ being an arbitrary vector
with appropriate length, we arrive at ${\bf{H}}_{aa}{\bf \Gamma}({\bf G}_{a}){\bf x}={\bf 0}$.
This is equivalent to ${\bf x}={\bf \Gamma}({\bf{H}}_{aa}{\bf \Gamma}({\bf G}_{a})){\bf y}$, with
$\bf y$ being an arbitrary vector with appropriate length.

Thus, the formula of ${\bf v}_a$ in $Sub_{11}$ is
\begin{align}
{\bf v}_a={\bf \Gamma}({\bf G}_{a}){\bf \Gamma}({\bf{H}}_{aa}{\bf \Gamma}({\bf G}_{a})){\bf z},  \label{eq10}
\end{align}
with $\bf z$ being any nonzero vectors with appropriate length.

2) $Sub_{12}$: The precoding vectors in $Sub_{12}$ should satisfy
\begin{subequations}
\begin{align}
&{\bf G}_a{\bf v}_a={\bf 0}, \label{eq12a} \\
&{\bf{H}}_{aa}{\bf v}_a\ne {\bf 0}. \label{eq12b}
\end{align}
\end{subequations}

The vectors ${\bf v}_a$ satisfying (\ref{eq12a}) are of the form ${\bf \Gamma}({\bf G}_{a}){\bf x}$.
Because ${\bf H}_{aa}$ is independent of ${\bf{G}}_{a}$, for precoding vectors satisfying (\ref{eq12a}),
${\bf H}_{aa}{\bf v}_a \ne 0$ holds true with probability one.
So, the vectors ${\bf v}_a$ in $Sub_{12}$ are of the form ${\bf \Gamma}({\bf G}_{a}){\bf x}$.

On the other hand, since we want linearly independent precoding vectors, the beamforming direction
already considered in the set with higher priority, e.g., $Sub_{11}$, should not be under consideration.
Thus, in $Sub_{12}$ we only consider the following vectors,
\begin{align}
{\bf v}_a={\bf \Gamma}({\bf G}_{a}){\bf \Gamma}^\perp({\bf{H}}_{aa}{\bf \Gamma}({\bf G}_{a})){\bf z}.  \label{eq13}
\end{align}


3) $Sub_{13}$: The precoding vectors in $Sub_{13}$ should satisfy
\begin{subequations}
\begin{align}
&{\bf G}_b{\bf v}_b={\bf 0},  \label{eq15a}\\
&{\bf{H}}_{bb}{\bf v}_b={\bf 0}. \label{eq15b}
\end{align}
\end{subequations}
In a similar way we derive the formula of ${\bf v}_a$ in $Sub_{11}$, we obtain
the formula of ${\bf v}_b$ in $Sub_{13}$, i.e.,
\begin{align}
{\bf v}_b={\bf \Gamma}({\bf G}_{b}){\bf \Gamma}({\bf{H}}_{bb}{\bf \Gamma}({\bf G}_{b})){\bf z}.  \label{eq16}
\end{align}

4) $Sub_{14}$: The precoding vectors in $Sub_{14}$ should satisfy
\begin{subequations}
\begin{align}
&{\bf G}_b{\bf v}_b={\bf 0}, \label{eq18a} \\
&{\bf{H}}_{bb}{\bf v}_b\ne {\bf 0}. \label{eq18b}
\end{align}
\end{subequations}
In a similar way we derive the formula of ${\bf v}_a$ in $Sub_{12}$, we obtain
the formula of ${\bf v}_b$ in $Sub_{14}$, i.e.,
\begin{align}
{\bf v}_b={\bf \Gamma}({\bf G}_{b}){\bf \Gamma}^\perp({\bf{H}}_{bb}{\bf \Gamma}({\bf G}_{b})){\bf z}.\label{eq19}
\end{align}

5) $Sub_{21}$: The precoding vector pairs in $Sub_{21}$ should satisfy
\begin{subequations}
\begin{align}
&{\bf{H}}_{bb}{\bf v}_b={\bf 0}, \label{eq21a} \\
&{\bf{H}}_{aa}{\bf v}_a= {\bf 0}, \label{eq21b}\\
&{\bf G}_a{\bf v}_a ={\bf G}_b{\bf v}_b \ne {\bf 0}. \label{eq21c}
\end{align}
\end{subequations}

Substituting ${\bf v}_a={\bf \Gamma}({\bf{H}}_{aa}){\bf x}$ and ${\bf v}_b={\bf \Gamma}({\bf{H}}_{bb}){\bf y}$
into (\ref{eq21c}) yields
\begin{align}
{\bf G}_a{\bf \Gamma}({\bf{H}}_{aa}){\bf x} = {\bf G}_b{\bf \Gamma}({\bf{H}}_{bb}){\bf y} \ne {\bf 0}. \label{eq22}
\end{align}

Via \emph{Proposition 1}(i) of \cite{Lingxiang162} we arrive at $\bf x$ and $\bf y$ satisfying (\ref{eq22}), i.e.,
\begin{align}
&{\bf x}=\hat{\bf \Psi}_{12}\hat{\bf \Lambda}_1^{-1}{\bf z}+{\bf\Gamma}({\bf{G}}_{a}{\bf \Gamma}({\bf H}_{aa})){\bf z}_a, \nonumber \\
&{\bf y}=\hat{\bf \Psi}_{22}\hat{\bf \Lambda}_2^{-1}{\bf z}+{\bf\Gamma}({\bf{G}}_{b}{\bf \Gamma}({\bf H}_{bb})){\bf z}_b, \nonumber
\end{align}
where ${\bf z}_a$ and ${\bf z}_b$ denote any vectors with appropriate length;
$\hat{\bf \Psi}_{12}$, $\hat{\bf \Lambda}_1$, $\hat{\bf \Psi}_{22}$, $\hat{\bf \Lambda}_2$, and $\hat s$
(to be used in the next subsection) correspond to the $ {\bf \Psi}_{12}$, $ {\bf \Lambda}_1$, $ {\bf \Psi}_{22}$, $ {\bf \Lambda}_2$ and $  s$,
and arise due to the GSVD of
$({\bf G}_a{\bf \Gamma}({\bf{H}}_{aa}))^H$ and $ ({\bf G}_b{\bf \Gamma}({\bf{H}}_{bb}))^H$.

Thus, the formulas for ${\bf v}_a$ and ${\bf v}_b$ in $Sub_{21}$ are of the form
\begin{subequations}
\begin{align}
&{\bf v}_a={\bf \Gamma}({\bf{H}}_{aa})\hat{\bf \Psi}_{12}\hat{\bf \Lambda}_1^{-1}{\bf z}+{\bf \Gamma}({\bf{H}}_{aa}){\bf\Gamma}({\bf{G}}_{a}{\bf \Gamma}({\bf H}_{aa})){\bf z}_a, \label{eq23a} \\
&{\bf v}_b={\bf \Gamma}({\bf{H}}_{bb})\hat{\bf \Psi}_{22}\hat{\bf \Lambda}_2^{-1}{\bf z}+{\bf \Gamma}({\bf{H}}_{bb}){\bf\Gamma}({\bf{G}}_{b}{\bf \Gamma}({\bf H}_{bb})){\bf z}_b. \label{eq23b}
\end{align}
\end{subequations}


6) $Sub_{22}$: The precoding vector pairs in $Sub_{22}$ should satisfy
\begin{subequations}
\begin{align}
&{\bf{H}}_{bb}{\bf v}_b={\bf 0},  \label{eq25a}\\
&{\bf{H}}_{aa}{\bf v}_a \ne {\bf 0}, \label{eq25b} \\
&{\bf G}_a{\bf v}_a ={\bf G}_b{\bf v}_b \ne {\bf 0}.\label{eq25c}
\end{align}
\end{subequations}

Substituting ${\bf v}_b={\bf \Gamma}({\bf{H}}_{bb}){\bf y}$
into (\ref{eq25c}), we arrive at
\begin{align}
{\bf G}_a{\bf v}_a = {\bf G}_b{\bf \Gamma}({\bf{H}}_{bb}){\bf y} \ne {\bf 0}. \label{eq26}
\end{align}

Via \emph{Proposition 1}(i) of \cite{Lingxiang162} we arrive at ${\bf v}_a$ and $\bf y$ satisfying (\ref{eq26}), i.e.,
\begin{align}
&{\bf v}_a=\bar{\bf \Psi}_{12}\bar{\bf \Lambda}_1^{-1}{\bf z}+{\bf\Gamma}({\bf{G}}_{a}){\bf z}_a, \nonumber \\
&{\bf y}=\bar{\bf \Psi}_{22}\bar{\bf \Lambda}_2^{-1}{\bf z}+{\bf\Gamma}({\bf{G}}_{b}{\bf \Gamma}({\bf H}_{bb})){\bf z}_b, \nonumber
\end{align}
where $\bar{\bf \Psi}_{12}$, $\bar{\bf \Lambda}_1$, $\bar{\bf \Psi}_{22}$, $\bar{\bf \Lambda}_2$, and $\bar s$
(to be used in the next subsection) correspond to the $ {\bf \Psi}_{12}$, $ {\bf \Lambda}_1$, $ {\bf \Psi}_{22}$, $ {\bf \Lambda}_2$ and $  s$,
and arise due to the GSVD of
${\bf G}_a^H$ and $ ({\bf G}_b{\bf \Gamma}({\bf{H}}_{bb}))^H$.

Thus, the formulas for ${\bf v}_a$ and ${\bf v}_b$ in $Sub_{22}$ are of the form
\begin{subequations}
\begin{align}
&{\bf v}_a=\bar{\bf \Psi}_{12}\bar{\bf \Lambda}_1^{-1}{\bf z}+{\bf\Gamma}({\bf{G}}_{a}){\bf z}_a, \label{eq27a} \\
&{\bf v}_b={\bf \Gamma}({\bf{H}}_{bb})\bar{\bf \Psi}_{22}\bar{\bf \Lambda}_2^{-1}{\bf z}+{\bf \Gamma}({\bf{H}}_{bb}){\bf\Gamma}({\bf{G}}_{b}{\bf \Gamma}({\bf H}_{bb})){\bf z}_b. \label{eq27b}
\end{align}
\end{subequations}
In the above, we used the fact that since ${\bf{H}}_{aa}$ is independent of
${\bf{H}}_{bb}$, ${\bf G}_a$ and ${\bf G}_b$, for the precoding vector pairs in (\ref{eq27a}) and (\ref{eq27b}),
${\bf{H}}_{aa}{\bf v}_a \ne {\bf 0}$ holds true with probability one.

%

7) $Sub_{23}$: The precoding vector pairs in $Sub_{23}$ should satisfy
\begin{subequations}
\begin{align}
&{\bf{H}}_{bb}{\bf v}_b\ne {\bf 0},  \label{eq29a}\\
&{\bf{H}}_{aa}{\bf v}_a = {\bf 0},  \label{eq29b}\\
&{\bf G}_a{\bf v}_a ={\bf G}_b{\bf v}_b \ne {\bf 0}. \label{eq29c}
\end{align}
\end{subequations}

Substituting ${\bf v}_a={\bf \Gamma}({\bf{H}}_{aa}){\bf x}$
into (\ref{eq29c}), we arrive at
\begin{align}
{\bf G}_a{\bf \Gamma}({\bf{H}}_{aa}){\bf x} = {\bf G}_b{\bf v}_b \ne {\bf 0}. \label{eq30}
\end{align}

Via \emph{Proposition 1}(i) of \cite{Lingxiang162} we arrive at $\bf x$ and ${\bf v}_b$ satisfying (\ref{eq30}), i.e.,
\begin{align}
&{\bf x}=\breve{\bf \Psi}_{12}\breve{\bf \Lambda}_1^{-1}{\bf z}+{\bf\Gamma}({\bf{G}}_{a}{\bf \Gamma}({\bf H}_{aa})){\bf z}_a, \nonumber \\
&{\bf v}_b=\breve{\bf \Psi}_{22}\breve{\bf \Lambda}_2^{-1}{\bf z}+{\bf\Gamma}({\bf{G}}_{b}){\bf z}_b. \nonumber
\end{align}
where $\breve{\bf \Psi}_{12}$, $\breve{\bf \Lambda}_1$, $\breve{\bf \Psi}_{22}$, $\breve{\bf \Lambda}_2$, and $\breve s$
(to be used in the next subsection) correspond to the $ {\bf \Psi}_{12}$, $ {\bf \Lambda}_1$, $ {\bf \Psi}_{22}$, $ {\bf \Lambda}_2$ and $  s$,
and arise due to the GSVD of
$({\bf G}_a{\bf \Gamma}({\bf{H}}_{aa}))^H$ and $ {\bf G}_b^H$.

Thus, the formulas for ${\bf v}_a$ and ${\bf v}_b$ in $Sub_{23}$ are of the form
\begin{subequations}
\begin{align}
&{\bf v}_a={\bf \Gamma}({\bf{H}}_{aa})\breve{\bf \Psi}_{12}\breve{\bf \Lambda}_1^{-1}{\bf z}+{\bf \Gamma}({\bf{H}}_{aa}){\bf\Gamma}({\bf{G}}_{a}{\bf \Gamma}({\bf H}_{aa})){\bf z}_a, \label{eq31a} \\
&{\bf v}_b=\breve{\bf \Psi}_{22}\breve{\bf \Lambda}_2^{-1}{\bf z}+{\bf\Gamma}({\bf{G}}_{b}){\bf z}_b. \label{eq31b}
\end{align}
\end{subequations}

%

8) $Sub_{24}$: The precoding vector pairs in $Sub_{24}$ should satisfy
\begin{subequations}
\begin{align}
&{\bf{H}}_{bb}{\bf v}_b \ne {\bf 0},  \label{eq33a}\\
&{\bf{H}}_{aa}{\bf v}_a \ne {\bf 0}, \label{eq33b}\\
&{\bf G}_a{\bf v}_a ={\bf G}_b{\bf v}_b \ne {\bf 0}. \label{eq33c}
\end{align}
\end{subequations}

Via \emph{Proposition 1}(i) of \cite{Lingxiang162} we arrive at that
the formulas for ${\bf v}_a$ and ${\bf v}_b$ in $Sub_{24}$ are of the form
\begin{subequations}
\begin{align}
&{\bf v}_a=\tilde{\bf \Psi}_{12}\tilde{\bf \Lambda}_1^{-1}{\bf z}+{\bf\Gamma}({\bf{G}}_{a}){\bf z}_a, \label{eq34a} \\
&{\bf v}_b=\tilde{\bf \Psi}_{22}\tilde{\bf \Lambda}_2^{-1}{\bf z}+{\bf\Gamma}({\bf{G}}_{b}){\bf z}_b, \label{eq34b}
\end{align}
\end{subequations}
where $\tilde{\bf \Psi}_{12}$, $\tilde{\bf \Lambda}_1$, $\tilde{\bf \Psi}_{22}$, $\tilde{\bf \Lambda}_2$, and $\tilde s$
(to be used in the next subsection) correspond to the $ {\bf \Psi}_{12}$, $ {\bf \Lambda}_1$, $ {\bf \Psi}_{22}$, $ {\bf \Lambda}_2$ and $  s$,
and arise due to the GSVD of
${\bf G}_a^H$ and $ {\bf G}_b^H$.

%

\subsection{The number of linearly independent candidate precoding vector pairs in each subset}
Since all the channel matrices are assumed to be full rank, and by (\ref{eq10})(\ref{eq13})(\ref{eq16})(\ref{eq19}), it respectively holds that
\begin{subequations}
\begin{align}
&d_{ 11} \le (N_a^t-N_e-N_a^r)^+, \label{eq35a} \\
&d_{ 12}\le \min \{N_a^r, (N_a^t-N_e)^+\}, \label{eq35b} \\
&d_{ 13}\le (N_b^t-N_e-N_b^r)^+, \label{eq35c} \\
&d_{ 14}\le \min \{N_b^r, (N_b^t-N_e)^+\}. \label{eq35d}
\end{align}
\end{subequations}

From the above subsection, one can see that the formulas for $ {\bf v}_a$ and $ {\bf v}_b $ in $Sub_{1i}$, $i=1,...,4$,
may have some common basis vectors with those in $Sub_{2j}$, $j=1,...,4$. For example,
some basis vectors of ${\bf v}_a$ in $Sub_{21}$, i.e., ${\bf \Gamma}({\bf{H}}_{aa}){\bf\Gamma}({\bf{G}}_{a}{\bf \Gamma}({\bf H}_{aa}))$,
also span the solution space of $Sub_{11}$, since they span the same space as
${\bf \Gamma}({\bf{G}}_{a}){\bf\Gamma}({\bf{H}}_{aa}{\bf \Gamma}({\bf G}_{a}))$; some basis vectors of ${\bf v}_b$ in $Sub_{21}$, i.e., ${\bf \Gamma}({\bf{H}}_{bb}){\bf\Gamma}({\bf{G}}_{b}{\bf \Gamma}({\bf H}_{bb}))$, also span the solution space of $Sub_{13}$, since they
span the same space as ${\bf \Gamma}({\bf{G}}_{b}){\bf\Gamma}({\bf{H}}_{bb}{\bf \Gamma}({\bf G}_{b}))$.
Since we are interested in linearly independent ${\bf v}_a$'s and ${\bf v}_b$'s, the number of linearly independent precoding vectors ${\bf v}_a$ and
${\bf v}_b$ considered in $Sub_{21}$ should be equal. Moreover, any paired selection of those basis vectors from $Sub_{21}$
cannot help increase the sum S.D.o.F., as compared with the case in which we respectively attribute those basis vectors to $Sub_{11}$ and $Sub_{ 13}$. Thus, for the pairs in $Sub_{21}$, we do not consider the basis vectors ${\bf \Gamma}({\bf{H}}_{aa}){\bf\Gamma}({\bf{G}}_{a}{\bf \Gamma}({\bf H}_{aa}))$ and ${\bf \Gamma}({\bf{H}}_{bb}){\bf\Gamma}({\bf{G}}_{b}{\bf \Gamma}({\bf H}_{bb}))$.
Therefore,  
\begin{align}
d_{ 21} \le \hat s. \label{eq36a}
\end{align}
The above arguments also apply to $Sub_{22}$, $Sub_{23}$ and $Sub_{24}$.
Thus, the equalities in (\ref{eq35a})-(\ref{eq35d}) hold true.

\begin{table}[!t]
\renewcommand{\arraystretch}{1.6}
\centering
\caption{The maximum number of linearly independent candidate precoding vector pairs.}
\begin{tabular}{|c|c|c|}
\hline
\textbf{Subsets}&\textbf{Maximum number
of linearly independent vectors} \\
\hline
{$Sub_{11}$}  &  $d_{ 11}=(N_a^t-N_e-N_a^r)^+$ \\
\hline
{$Sub_{12}$}  &  $d_{ 12}= \min \{N_a^r, (N_a^t-N_e)^+\}$ \\
\hline
{$Sub_{13}$}  &  $d_{ 13}=(N_b^t-N_e-N_b^r)^+$ \\
\hline
{$Sub_{14}$}  &  $d_{ 14}= \min \{N_b^r, (N_b^t-N_e)^+\}$ \\
\hline
{$Sub_{21}$}  &  $d_{ 21}= \hat s $\\        
\hline
{$Sub_{22}$}  &  $d_{ 22}=\bar s -d_{ 21}$ \\        
\hline
{$Sub_{23}$}  &  $d_{ 23}=\breve s -d_{ 21}$  \\  
\hline
{$Sub_{24}$}  &  $d_{ 24}=\tilde s-(d_{ 21}+d_{ 22}+d_{ 23})$    \\ 
\hline
\end{tabular}
\end{table}

On the other hand, on combining (\ref{eq21a})-(\ref{eq21c}) with (\ref{eq25a})-(\ref{eq25c}), it holds that
\begin{align}
Sub_{21} \cup Sub_{22}=\{({\bf v}_a, {\bf v}_b)|{\bf{H}}_{bb}{\bf v}_b={\bf 0},
{\bf{G}}_{a}{\bf v}_a ={\bf{G}}_{b}{\bf v}_b \ne {\bf 0}\}. \nonumber
\end{align}
Thus,
\begin{align}
d_{ 22} +d_{ 21}\le \bar s. \label{eq36b}
\end{align}
On combining (\ref{eq21a})-(\ref{eq21c}) with (\ref{eq29a})-(\ref{eq29c}), it holds that
\begin{align}
Sub_{21} \cup Sub_{23}=\{({\bf v}_a, {\bf v}_b)|{\bf{H}}_{aa}{\bf v}_a = {\bf 0},
{\bf{G}}_{a}{\bf v}_a ={\bf{G}}_{b}{\bf v}_b \ne {\bf 0}\}. \nonumber
\end{align}
Thus,
\begin{align}
&d_{ 23} +d_{ 21}\le \breve s. \label{eq36c}
\end{align}
On combining (\ref{eq21a})-(\ref{eq21c}) with (\ref{eq25a})-(\ref{eq25c}), (\ref{eq29a})-(\ref{eq29c}) and
(\ref{eq33a})-(\ref{eq33c}), it holds that
$ Sub_{21} \cup Sub_{22} \cup Sub_{23} \cup Sub_{24}=\{({\bf v}_a, {\bf v}_b)|
{\bf{G}}_{a}{\bf v}_a ={\bf{G}}_{b}{\bf v}_b \ne {\bf 0}\} $.
Thus,
\begin{align}
d_{ 24} +d_{ 23}+d_{ 22} +d_{ 21}\le \tilde s. \label{eq36d}
\end{align}
Regarding the priority of $Sub_{2j}$, $j=1,...,4$, $Sub_{21}$ ranks the highest
and $Sub_{24}$ ranks the lowest.
Therefore, all the inequalities in (\ref{eq36a})-(\ref{eq36d}) hold true.

Based on the above discussions, and using (\ref{eq001d}),
Table I provides the number of linearly independent vectors $({\bf v}_a, {\bf v}_b)$'s that should be
considered in each subset.

\section{Construction of $({\bf V}_a, {\bf V}_b)$ Which Achieves the Maximum Sum S.D.o.F.}
The key idea for achieving the maximum sum S.D.o.F. is to include as many interference free
precoding vector pairs (interference free signal streams) in $({\bf V}_a, {\bf V}_b)$ as possible. To achieve that goal,
in the construction of $({\bf V}_a, {\bf V}_b)$,
we will select as many precoding vector pairs $({\bf v}_a, {\bf v}_b)$
from the subset with higher priority as possible. When there are no more available pairs in
a given subset, we will consider the next subset in terms of priority.

By the definition of each subset and the equations in (\ref{eqc1a}) and (\ref{eqc1b}),
one can see that $Sub_{21}$ has the highest priority, followed by
$Sub_{11} \cup Sub_{13} \cup Sub_{22}\cup Sub_{23}$, and then
$Sub_{24} \cup Sub_{12}\cup Sub_{14}$, i.e.,
\begin{align}
Sub_{21} \succ Sub_{11} \cup Sub_{13} &\cup Sub_{22}\cup Sub_{23} \succ Sub_{\rm L},  \label{eqrank}
\end{align}
where $Sub_{\rm L} \triangleq Sub_{24} \cup Sub_{12}\cup Sub_{14}$.
Moreover, $Sub_{12}$ and $ Sub_{14}$ have the same priority.
$Sub_{24}$ has higher priority than $Sub_{12}\cup Sub_{14}$ except for the case in which
the dimensions of the available interference free receive subspace \emph{Alice} and \emph{Bob} can see are both equal to one.
The priorities of $Sub_{11} $, $ Sub_{13} $, $ Sub_{22}$ and $Sub_{23}$ depend on the number of antennas,
and also the available interference free receive subspace \emph{Alice} and \emph{Bob} can respectively see.

In the following, we will consider four distinct cases, i.e., the case of
$N_a^t \le N_e+N_a^r$ and $N_b^t\le N_e+N_b^r$, the case of $N_a^t \le N_e+N_a^r$ and $N_b^t > N_e+N_b^r$,
the case of $N_a^t > N_e+N_a^r$ and $N_b^t \le N_e+N_b^r$, and the case of $N_a^t \le N_e+N_a^r$ and $N_b^t > N_e+N_b^r$.
For each case, we will determine the specific rankings of $Sub_{11} $, $ Sub_{13} $, $ Sub_{22}$ and $Sub_{23}$.

By those rankings, we include precoding vector pairs into the precoding matrix pair $({\bf V}_a, {\bf V}_b)$ one by one,
until adding more precoding vector pairs does not increase the achievable sum S.D.o.F..
With this constructive method and stop criteria,
we then determine the maximum achievable sum S.D.o.F..

\subsection{$N_a^t \le N_e+N_a^r$, $N_b^t\le N_e+N_b^r$.}
\begin{figure*}[!t]
\normalsize
\setcounter{equation}{36}
{\small\begin{align}
&q_1= {\min}^+\{\textcolor{blue}{  \max}\{N_b^r,N_a^r\}, d_{21}\}, \nonumber  \\
&q_2= {\min}^+\{\textcolor{blue}{  \max}\{\lfloor\frac {N_b^r-q_1}{2}\rfloor,N_a^r-q_1\}, \min\{N_b^r-N_a^r, d_{23}\}\}, \nonumber  \\   
&q_3= 2{\min}^+ \{\textcolor{blue}{  \max}\{\lfloor\frac {N_b^r-q_1-2q_2}{3}\rfloor,\lfloor\frac {N_a^r-q_1-q_2}{3}\rfloor\}, \bar d_{23}\}, \nonumber \\
&q_4= {\min}^+ \{\textcolor{blue}{  \max}\{N_b^r-q_1-2q_2-3q_3,\lfloor\frac {N_a^r-q_1-q_2-3q_3}{2}\rfloor\},d_{22}-q_3\}, \nonumber\\
&q_5= {\min}^+ \{\textcolor{blue}{  \max}\{\lfloor\frac {N_b^r-q_1-2q_2-3q_3-q_4}{2}\rfloor,\lfloor\frac {N_a^r-q_1-q_2-3q_3-2q_4}{2}\rfloor\},d_{24}\}, \nonumber\\
&q_6={\min}^+\{\textcolor{blue}{  \min}\{N_b^r-q_1-2q_2-3q_3-q_4-2q_5,N_a^r-q_1-q_2-3q_3-2q_4-2q_5\}, d_{12} \}, \nonumber\\
&{\small q_7={\min}^+\{\textcolor{blue}{  \min}\{N_b^r-q_1-2q_2-3q_3-q_4-2q_5-q_6,N_a^r-q_1-q_2-3q_3-2q_4-2q_5-q_6\}, d_{14} \}}. \label{eqqn}
\end{align}}
\setcounter{equation}{37}
\hrulefill
\end{figure*}

In this case, and based on Table I it holds that $d_{11}=d_{13}=0$. Thus, we only need to consider the rankings of
$Sub_{22}$ and $Sub_{23}$. By definition, for the precoding vector pairs from $Sub_{22}$
\emph{Bob} is self-interference free while \emph{Alice} suffers from self-interference; for
the precoding vector pairs from $Sub_{23}$ the situation reverses. Combined with
(\ref{eqc1a}) and (\ref{eqc1b}), one can see that if
\emph{Alice} has a greater interference free receive subspace,
$Sub_{22}$ is of higher priority; otherwise, $Sub_{23}$ is of higher priority.
For example, consider the case $N_b^r=1$ and $N_a^r=2$, selecting
one pair from $Sub_{22}$ achieves a sum S.D.o.F. of 2, while selecting one pair
from $Sub_{23}$ can only achieve a sum S.D.o.F. of 1.

So far, we have obtained the specific rankings for all the subsets, based on which,
we include precoding vector pairs into the precoding matrix pair $({\bf V}_a, {\bf V}_b)$ one by one,
until the further adding of precoding vector pairs could not help increase the achievable sum S.D.o.F..
This constructive method also provides us a way to determine the maximum achievable S.D.o.F..
In the following, we will consider two distinct subcases,
i.e., the subcase of $N_b^r \ge N_a^r$ and the subcase of $N_b^r < N_a^r$.

i) \emph{For the subcase of $N_b^r \ge N_a^r$.}

Following (\ref{eqrank}) and the rankings discussed above, we first
pick pairs from $Sub_{21}$. Until there are no available pairs in $Sub_{21}$, we next consider pairs from $Sub_{22}$ and $Sub_{23}$.
$Sub_{23}$ has higher priority and its pairs are
at the first consideration, since \emph{Bob} can see a bigger interference free receive subspace.
However, as we include one pair from $Sub_{23}$,
the dimension of the available interference free receive subspace \emph{Bob} and \emph{Alice} respectively decreases by two and one.
After selecting $N_b^r - N_a^r+1$ pairs from $Sub_{23} $, 
\emph{Alice} has a greater interference free receive subspace, which indicates that
$Sub_{22}$ has higher priority. Similarly, as we include one pair from $Sub_{22}$,
the dimension of the available interference free receive subspace that \emph{Alice} and \emph{Bob} can see
respectively decreases by one and two.
Thus, following the selection of one pair from $Sub_{23}$, the pairs from $Sub_{22}$ have higher priority. Summarizing the
above observations, after picking all the pairs from $Sub_{21}$, we will first select
$\min \{d_{23}, N_b^r - N_a^r\}$ pairs from $Sub_{23} $; we then select
pairs from $Sub_{22} $ and $ Sub_{23}$ in turn, one by one,
until there are no more available pairs in $Sub_{22}$ or $ Sub_{23}$. Subsequently, we select pairs from
the remaining pairs from $Sub_{22} $ or $ Sub_{23}$, followed by $Sub_{\rm L}$.

Based on the above constructive method, we now can determine the maximum achievable S.D.o.F..
Divide $Sub_{23}$ into two subsets, i.e., $Sub_{23}^1$ and $Sub_{23}^2$,
with a number of $\min \{d_{23}, N_b^r - N_a^r\}$ and
$\bar d_{23} \triangleq d_{23}-\min \{d_{23}, N_b^r - N_a^r\}$ precoding vector pairs, respectively.
Assume that $d_{ 22} \ge \bar d_{23} $. Then,
we will first run out of pairs of $Sub_{23}^2$.
Therefore, in the construction of $({\bf V}_a, {\bf V}_b)$,
the subsets are ranked as
\begin{align}
Sub_{21} \succ Sub_{23}^1 \succ \{Sub_{22}, Sub_{23}^2\} \succ
Sub_{22} \succ Sub_{\rm L}, \nonumber
\end{align}
where by $\{Sub_{22}, Sub_{23}^2\}$ we
mean that we select pairs from $Sub_{22} $ and $ Sub_{23}$ in turn.
Let $q_1$, $q_2$, $q_3$, $q_4$, $q_5$, $q_6$ and $q_7$ be the number of precoding vector pairs we pick from $Sub_{21}$, $Sub_{23}^1$,
$ \{Sub_{22}, Sub_{23}^2\}$, $Sub_{22}$, $Sub_{24}$, $Sub_{12} $ and $Sub_{14}$, respectively.
Since we stop picking until the further adding of precoding vector pairs could not help increase the achievable sum S.D.o.F.,
the expressions of $q_1$, $q_2$, $q_3$, $q_4$, $q_5$, $q_6$ and $q_7$ can be written as in (\ref{eqqn}) at the top of this page.

According to (\ref{eqc1a}), the achieved S.D.o.F. of the \emph{Alice}-\emph{Bob} channel is
\begin{align}
&d_{\rm I}^a=\min\{(N_b^r-q_2-\frac{q_3}{2}-q_5-q_7)^+,\sum\limits_{i=1}^6 q_i\}. \nonumber
\end{align}
According to (\ref{eqc1b}), the achieved S.D.o.F. of the \emph{Bob}-\emph{Alice} channel is
\begin{align}
&d_{\rm I}^b=\min\{(N_a^r-q_2-\frac{q_3}{2}-q_4-q_5-q_6)^+,(\sum\limits_{i=1}^7 q_i)-q_6\}. \nonumber
\end{align}
Therefore, the maximum achievable sum S.D.o.F. is
\begin{align}
&d_s^{\rm sum}=d_{\rm I}^b+d_{\rm I}^a.\label{eq40}
\end{align}

\emph{Example 1:} Consider the case $(N_a^t, N_a^r)=(5, 2)$, $(N_b^t, N_b^r)=(4, 3)$ and $N_e=5$.
Based on Table I, the maximum number of linearly independent precoding vector pairs in each subset is
$d_{ 11}=d_{ 12}=d_{ 13}=d_{ 14}=d_{ 21}=0$, $d_{ 22}=1$, $d_{ 23}=2$, $d_{ 24}=1$.
Since $N_b^r \ge N_a^r$, we divide $Sub_{23}$ into two subsets, i.e., $Sub_{23}^1$ and $Sub_{23}^2$,
with each consisting of one precoding vector pair.
By the rankings derived above, i.e., $Sub_{23}^1 \succ \{Sub_{22}, Sub_{23}^2\} \succ
Sub_{22} \succ Sub_{24}$, we first select a precoding vector pair
from $Sub_{23}^1$, i.e, $({\bf v}_a^1, {\bf v}_b^1)$. The dimension of the remaining available interference free receive subspace
of \emph{Alice} and \emph{Bob} are $N_a^r-2=1$ and $N_b^r-1=1$, respectively.
For the remaining precoding vector pairs, we stop selecting after we take one more from
$Sub_{22}$, i.e., $({\bf v}_a^2, {\bf v}_b^2)$, since
additional precoding vector pairs will introduce extra
CCI without increasing the sum S.D.o.F.. For $Sub_{23}^1 $, it holds that ${\bf H}_{aa}{\bf v}_a^1 = {\bf 0}$.
for $Sub_{22}$, it holds that ${\bf H}_{bb}{\bf v}_b^2 = {\bf 0}$.
Therefore, for ${\bf V}_a=[{\bf v}_a^1, {\bf v}_a^2]$ and ${\bf V}_b=[{\bf v}_b^1, {\bf v}_b^2]$, with (\ref{eqc1a}) and (\ref{eqc1b})
it holds that $d_s^a=1$ and $d_s^b=2$. Concluding,
a sum S.D.o.F. of 3 can be achieved.

The precoding vector pair selection procedure for the case of
$d_{ 22} \le \bar d_{23} $ is similar to that for the case of $d_{ 22} > \bar d_{23} $.
Moreover, the maximum achievable sum S.D.o.F. can be obtained by (\ref{eq40}),
with the place of $ \bar d_{23} $ and $d_{ 22} $ in (\ref{eqqn}) exchanged.

ii) The discussion for the subcase of $N_b^r < N_a^r$
is omitted since it is similar to that for the subcase of $N_b^r \ge N_a^r$.

\subsection{$N_a^t \le N_e+N_a^r$, $N_b^t > N_e+N_b^r$.}
In this case,  and based on Table I it holds that $d_{11}=d_{23}=d_{24}=0$.
Thus, we only need to consider the rankings of $Sub_{13}$ and $Sub_{22}$.
By definition, for the precoding vector pairs from $Sub_{13}$, ${\bf v}_a={\bf 0}$ and
\emph{Bob} is self-interference free; for the precoding vector pairs from $Sub_{22}$,
\emph{Bob} is self-interference free while \emph{Alice} suffers from self-interference.
Combined with (\ref{eqc1a}) and (\ref{eqc1b}), one can see that
$Sub_{22}$ has higher priority than $Sub_{13}$ except for the case in which
the dimension of the available interference free receive subspace \emph{Bob} can see is zero.
Divide $Sub_{22}$ into two subsets, i.e., $Sub_{22}^1$ and $Sub_{22}^2$,
with a number of $\min \{d_{22}, N_b^r\}$ and
$\hat d_{22} \triangleq d_{22}-\min \{d_{22}, N_b^r \}$ precoding vector pairs, respectively.
Based on what discussed above, we consider the following rankings:
\begin{align}
Sub_{21} \succ Sub_{22}^1 \succ Sub_{13}\succ Sub_{22}^2 \succ Sub_{12}=Sub_{14}. \nonumber
\end{align}

Let $\zeta_1$, $\zeta_2$, $\zeta_3$, $\zeta_4$, $\zeta_5$ and $\zeta_6$ be the number of precoding vector pairs we pick from
$Sub_{21}$, $ Sub_{22}^1$, $ Sub_{13}$, $ Sub_{22}^2$, $Sub_{12}$ and $Sub_{14}$, respectively.
We stop picking until the further adding of precoding vector pairs could not help increase the achievable sum S.D.o.F..
Then, the expressions of $\zeta_1$, $\zeta_2$, $\zeta_3$, $\zeta_4$, $\zeta_5$ and $\zeta_6$ are as follows:
\begin{align}
&\zeta_1= {\min}^+ \{\textcolor{blue}{  \max}\{N_b^r, N_a^r\}, d_{21} \}, \nonumber  \\  
&\zeta_2= {\min}^+ \{ \textcolor{blue}{  \max} \{N_b^r-\zeta_1, \lfloor\frac {N_a^r-\zeta_1}{2}\rfloor\}, \min \{d_{22}, N_b^r\}\}, \nonumber \\
&\zeta_3= {\min}^+ \{N_a^r-\zeta_1-2\zeta_2, d_{13}\},\nonumber \\
&\zeta_4= {\min}^+ \{\textcolor{blue}{  \max}\{N_b^r-\zeta_1-\zeta_2,\lfloor\frac {N_a^r-\zeta_1-2\zeta_2-\zeta_3}{2}\rfloor\}, \hat d_{22}\}, \nonumber \\
&\zeta_5= {\min}^+ \{\textcolor{blue}{  \min}\{N_b^r-\zeta_1-\zeta_2-\zeta_4,\bar N_a^r\}, d_{12}\}, \nonumber \\
&\zeta_6= {\min}^+ \{\textcolor{blue}{  \min}\{N_b^r-\zeta_1-\zeta_2-\zeta_4-\zeta_5,\bar N_a^r-\zeta_5\}, d_{14}\}, \nonumber 
\end{align}
with $\bar N_a^r=N_a^r-\zeta_1-2\zeta_2-\zeta_3-2\zeta_4$.

According to (\ref{eqc1a}), the achieved S.D.o.F. of the \emph{Alice}-\emph{Bob} channel is
\begin{align}
&d_{\rm II}^a=\min\{(N_b^r-\zeta_6)^+,\zeta_1+\zeta_2+\zeta_4+\zeta_5\}, \nonumber
\end{align}
According to (\ref{eqc1b}), the achieved S.D.o.F. of the \emph{Bob}-\emph{Alice} channel is
\begin{align}
&d_{\rm II}^b=\min\{(N_a^r-\zeta_2-\zeta_4-\zeta_5)^+,\zeta_1+\zeta_2+\zeta_3+\zeta_4+\zeta_6\}.\nonumber
\end{align}
Therefore, the maximum achievable sum S.D.o.F. is
\begin{align}
&d_s^{\rm sum}=d_{\rm II}^b+d_{\rm II}^a. \label{eq43}
\end{align}

\emph{Example 2:} Consider the case $(N_a^t, N_a^r)=(4, 6)$, $(N_b^t, N_b^r)=(8, 2)$ and $N_e=5$.
Based on Table I, the number of candidate precoding vector pairs in each subset is
$d_{ 11}=d_{ 12}=d_{ 21}=d_{ 23}=d_{ 24}=0$, $d_{ 13}=1$, $d_{ 14}=2$, $d_{ 22}=4$.
Since $ Sub_{22}^1 \succ Sub_{13} \succ Sub_{22}^2 \succ  Sub_{14}$, the precoding vector pairs of $Sub_{22}^1$ are
at the first consideration. Thus, we first select two precoding vector pairs
from $Sub_{22}^1$, i.e, $({\bf v}_a^1, {\bf v}_b^1)$ and $({\bf v}_a^2, {\bf v}_b^2)$.
The dimension of the remaining available interference free receive subspace
of \emph{Alice} and \emph{Bob} are $N_a^r-4=2$ and $N_b^r-2=0$, respectively.
For the remaining pairs, we stop after we take one more from $Sub_{13}$, i.e.,
$({\bf v}_a^3, {\bf v}_b^3)$, since additional precoding vector will introduce extra CCI without increasing the sum S.D.o.F..
For $Sub_{22}$, it holds that ${\bf H}_{bb}{\bf v}_b^1 = {\bf 0}$ and ${\bf H}_{bb}{\bf v}_b^2 = {\bf 0}$;
for $Sub_{13}$, it holds that ${\bf H}_{bb}{\bf v}_b^3 = {\bf 0}$.
Therefore, for ${\bf V}_a=[{\bf v}_a^1, {\bf v}_a^2, {\bf v}_a^3]$ and
${\bf V}_b=[{\bf v}_b^1, {\bf v}_b^2, {\bf v}_b^3]$, with (\ref{eqc1a}) and (\ref{eqc1b})
it holds that $d_s^a=2$ and $d_s^b=3$. Concluding,
a sum S.D.o.F. of 5 can be achieved.

\begin{table*}[!htp]
\caption{An algorithm for constructing $({\bf V}_a, {\bf V}_b)$ which achieves the maximum sum S.D.o.F..}
 \begin{tabular}{p{0.95\linewidth}}
  \hline
\textbf{Step 1}. {Initialization}: compute $d_{ij}$, $i=1,2$, $j=1,\cdots,4$, according to Table I;
compute the precoding vector pairs $({\bf v}_a, {\bf v}_b)$ for $Sub_{ij}$, $i=1,2$, $j=1,\cdots,4$, with (\ref{eq10})(\ref{eq13})
(\ref{eq16})(\ref{eq19})(\ref{eq23a})(\ref{eq23b})(\ref{eq27a})(\ref{eq27b})(\ref{eq31a})(\ref{eq31b})(\ref{eq34a})(\ref{eq34b}), respectively.
 \\\\
\textbf{Step 2}. Rankings for different number of antennas: \\
\emph{Case A: $N_a^t \le N_e+N_a^r$ and $N_b^t\le N_e+N_b^r$. }\\
\quad i) If $N_b^r \ge N_a^r$, divide $Sub_{23}$ into $Sub_{23}^1$ and $Sub_{23}^2$,
with a number of $\min \{d_{23}, N_b^r - N_a^r\}$ and $\bar d_{23} \triangleq d_{23}-\min \{d_{23}, N_b^r - N_a^r\}$ \\
\quad precoding vector pairs, respectively.\\
\quad\quad a) If $d_{ 22} \ge \bar d_{23} $, the subsets are ranked as
$Sub_{21} \succ Sub_{23}^1 \succ \{Sub_{22}, Sub_{23}^2\} \succ
Sub_{22} \succ Sub_{24} \cup Sub_{12}\cup Sub_{14}$. \\ 
\quad\quad b) If $d_{ 22} < \bar d_{23} $, the subsets are ranked as $Sub_{21} \succ Sub_{23}^1 \succ \{Sub_{22}, Sub_{23}^{2}\} \succ
Sub_{23}^{2} \succ Sub_{24} \cup Sub_{12}\cup Sub_{14}$. 
\\
\quad ii) If $N_b^r < N_a^r$, divide $Sub_{22}$ into $Sub_{22}^1$ and $Sub_{22}^2$,
with a number of $\min \{d_{22}, N_a^r - N_b^r\}$ and $\bar d_{22} \triangleq d_{22}-\min \{d_{22}, N_a^r - N_b^r\}$ \\
\quad precoding vector pairs, respectively.\\
\quad\quad a) If $d_{ 23} \ge \bar d_{22} $, the subsets are ranked as
$Sub_{21} \succ Sub_{22}^1 \succ \{Sub_{23}, Sub_{22}^2\} \succ
Sub_{23} \succ Sub_{24} \cup Sub_{12}\cup Sub_{14}$. 
\\
\quad\quad b) If $d_{ 23} < \bar d_{22} $, the subsets are ranked as $Sub_{21} \succ Sub_{22}^1 \succ \{Sub_{23}, Sub_{22}^{2}\} \succ
Sub_{22}^{2} \succ Sub_{24} \cup Sub_{12}\cup Sub_{14}$. 
\\
\emph{Case B: $N_a^t \le N_e+N_a^r$, $N_b^t > N_e+N_b^r$. }\\
\quad Divide $Sub_{22}$ into two subsets, i.e., $Sub_{22}^1$ and $Sub_{22}^2$,
with a number of $\min \{d_{22}, N_b^r\}$ and
$\hat d_{22} \triangleq d_{22}-\min \{d_{22}, N_b^r \}$ precoding vector pairs, respectively.
The subsets are ranked as
$Sub_{21}  \succ Sub_{22}^1 \succ Sub_{13} \succ Sub_{22}^2 \succ Sub_{12}=Sub_{14}$.
\\
\emph{Case C: $N_a^t > N_e+N_a^r$, $N_b^t \le N_e+N_b^r$. }\\
\quad Divide $Sub_{23}$ into two subsets, i.e., $Sub_{23}^1$ and $Sub_{23}^2$,
with a number of $\min \{d_{23}, N_a^r\}$ and
$\hat d_{23} \triangleq d_{23}-\min \{d_{23}, N_a^r \}$ precoding vector pairs, respectively.
The subsets are ranked as
$Sub_{21} \succ  Sub_{23}^1 \succ Sub_{11} \succ  Sub_{23}^2 \succ Sub_{12}=Sub_{14}$.
\\
\emph{Case D: $N_a^t > N_e+N_a^r$, $N_b^t > N_e+N_b^r$. }\\
\quad The subsets are ranked as
$Sub_{21} \succ Sub_{11}=Sub_{13}  \succ Sub_{12}=Sub_{14}$.
\\ \\
\textbf{Step 3}. Include precoding vector pairs into $({\bf V}_a, {\bf V}_b)$ one by one from $Sub_{ij}$, $i=1,2$, $j=1,\cdots,4$,
according to the rankings listed in Step 2.
Stop when adding more precoding vector pairs does increase the achievable sum S.D.o.F..
\\\\
\textbf{Step 4}. {Output:} $({\bf V}_a, {\bf V}_b)$.
\\

\hline
\end{tabular}
\end{table*}

\subsection{$N_a^t > N_e+N_a^r$, $N_b^t\le N_e+N_b^r$.}
In this case, and based on Table I it holds that $d_{13}=d_{22}=d_{24}=0$.
Thus, we only need to consider the rankings of $Sub_{11}$ and $Sub_{23}$.
By definition, for the precoding vector pairs from $Sub_{11}$, ${\bf v}_b={\bf 0}$ and
\emph{Alice} is self-interference free; for the precoding vector pairs from $Sub_{23}$,
\emph{Alice} is self-interference free while \emph{Bob} suffers from self-interference.
Combined with (\ref{eqc1a}) and (\ref{eqc1b}), one can see that
$Sub_{23}$ has higher priority than $Sub_{11}$ except for the case in which
the dimension of the available interference free receive subspace \emph{Alice} can see is zero.
Divide $Sub_{23}$ into two subsets, i.e., $Sub_{23}^1$ and $Sub_{23}^2$,
with a number of $\min \{d_{23}, N_a^r\}$ and
$\hat d_{23} \triangleq d_{23}-\min \{d_{23}, N_a^r \}$ precoding vector pairs, respectively.
Based on what discussed above, we consider the following rankings:
\begin{align}
Sub_{21}\succ Sub_{23}^1 \succ Sub_{11} \succ Sub_{23}^2 \succ Sub_{12}=Sub_{14}. \nonumber
\end{align}

Let $\eta_1$, $\eta_2$, $\eta_3$, $\eta_4$, $\eta_5$ and $\eta_6$ be the number of precoding vector pairs we pick from
$Sub_{21}$, $ Sub_{23}^1$, $ Sub_{11}$, $ Sub_{23}^2$, $Sub_{12}$ and $Sub_{14}$, respectively.
We stop picking until the further adding of precoding vector pairs could not help increase the achievable sum S.D.o.F..
Then, the expressions of $\eta_1$, $\eta_2$, $\eta_3$, $\eta_4$, $\eta_5$ and $\eta_6$ are as follows:
\begin{align}
&\eta_1= {\min}^+ \{\textcolor{blue}{  \max}\{N_a^r, N_b^r\}, d_{21} \}, \nonumber  \\  
&\eta_2={\min}^+ \{\textcolor{blue}{  \max}\{\lfloor\frac {N_b^r-\eta_1}{2}\rfloor, N_a^r-\eta_1\}, \min \{d_{23}, N_a^r\}\}, \nonumber \\
&\eta_3= {\min}^+ \{N_b^r-\eta_1-2\eta_2, d_{11}\}, \nonumber \\
&\eta_4= {\min}^+ \{\textcolor{blue}{  \max}\{\lfloor\frac {N_b^r-\eta_1-2\eta_2-\eta_3}{2}\rfloor, N_a^r-\eta_1-\eta_2\}, \hat d_{23}\}, \nonumber\\
&\eta_5= {\min}^+ \{\textcolor{blue}{  \min}\{\bar N_b^r, N_a^r-\eta_1-\eta_2-\eta_4\}, d_{12}\}, \nonumber\\
&\eta_6= {\min}^+ \{\textcolor{blue}{  \min}\{\bar N_b^r-\eta_5, N_a^r-\eta_1-\eta_2-\eta_4-\eta_5\}, d_{14}\}, \nonumber 
\end{align}
with $\bar N_b^r=N_b^r-\eta_1-2\eta_2-\eta_3-2\eta_4$.

According to (\ref{eqc1a}), the achieved S.D.o.F. of the \emph{Alice}-\emph{Bob} channel is
\begin{align}
&d_{\rm III}^a=\min\{(N_b^r-\eta_2-\eta_4-\eta_6)^+,\eta_1+\eta_2+\eta_3+\eta_4+\eta_5\}.  \nonumber
\end{align}
According to (\ref{eqc1b}), the achieved S.D.o.F. of the \emph{Bob}-\emph{Alice} channel is
\begin{align}
&d_{\rm III}^b=\min\{(N_a^r-\eta_5)^+,\eta_1+\eta_2+\eta_4+\eta_6\}. \nonumber
\end{align}
Therefore, the maximum achievable sum S.D.o.F. is
\begin{align}
&d_s^{\rm sum}=d_{\rm III}^b+d_{\rm III}^a. \label{eq46}
\end{align}

%
%

\subsection{$N_a^t > N_e+N_a^r$, $N_b^t > N_e+N_b^r$.}
In this case, and based on Table I it holds that $d_{22}=d_{23}=d_{24}=0$.
Thus, we only need to consider the rankings of $Sub_{11}$ and $Sub_{13}$.
By definition, for the precoding vector pairs in $Sub_{{11}}$, ${\bf v}_b= {\bf 0}$ and \emph{Alice} is self-interference free;
for the precoding vector pairs in $Sub_{{13}}$, ${\bf v}_a={\bf 0}$ and \emph{Bob} is self-interference free.
Therefore, the relative rankings of $Sub_{11}$ and $Sub_{13}$ will not affect
the achievable sum S.D.o.F.. In the construction of $({\bf V}_a, {\bf V}_b)$, the subsets are ranked as
\begin{align}
Sub_{21} \succ Sub_{11}=Sub_{13}  \succ Sub_{12}=Sub_{14}. \nonumber
\end{align}

Let $t_1$, $t_2$, $t_3$, $t_4$ and $t_5$ be the number of precoding vector pairs we pick from
$Sub_{21} $, $Sub_{11}$, $ Sub_{13}$, $Sub_{12}$ and $ Sub_{14}$, respectively.
We stop picking until the further adding of precoding vector pairs could not help increase the achievable sum S.D.o.F..
Then, the expressions of $t_1$, $t_2$, $t_3$, $t_4$ and $t_5$ are as follows:
\begin{align}
&t_1= {\min}^+ \{\textcolor{blue}{  \max}\{N_a^r, N_b^r\}, d_{21} \}, \nonumber  \\
&t_2= {\min}^+ \{ { N_b^r-t_1}, d_{11}\}, \nonumber \\
&t_3= {\min}^+ \{ { N_a^r-t_1}, d_{13}\}, \nonumber \\
&t_4= {\min}^+ \{\textcolor{blue}{  \min}\{ N_b^r-t_1-t_2, N_a^r-t_1-t_3\}, d_{12}\}, \nonumber \\
&t_5= {\min}^+ \{\textcolor{blue}{  \min}\{ N_b^r-t_1-t_2-t_4, N_a^r-t_1-t_3-t_4\}, d_{14}\}.  \nonumber 
\end{align}

According to (\ref{eqc1a}), the achieved S.D.o.F. of the \emph{Alice}-\emph{Bob} channel is
\begin{align}
&d_{\rm IV}^a=\min\{(N_b^r-t_5)^+,t_1+t_2+t_4\}. \nonumber
\end{align}
According to (\ref{eqc1b}), the achieved S.D.o.F. of the \emph{Bob}-\emph{Alice} channel is
\begin{align}
&d_{\rm IV}^b=\min\{(N_a^r-t_4)^+,t_1+t_3+t_5\}.  \nonumber
\end{align}
Therefore, the maximum achievable sum S.D.o.F. is
\begin{align}
&d_s^{\rm sum}=d_{\rm IV}^b+d_{\rm IV}^a. \label{eq49}
\end{align}

\emph{Example 3:} Consider the case $(N_a^t, N_a^r)=(7, 4)$, $(N_b^t, N_b^r)=(7, 4)$ and $N_e=2$.
Based on Table I, the maximum number of linearly independent precoding vector pairs in each subset is
$d_{ 11}=d_{ 13}=1$, $d_{ 12}=d_{ 14}=4$, $d_{ 21}=2$, $d_{ 22}=d_{ 23}=d_{ 24}=0$.
By the rankings $Sub_{21} \succ Sub_{11}=Sub_{13}  \succ Sub_{12}=Sub_{14}$,
we first select two precoding vector pairs
from $Sub_{21}$, one 
precoding vector pair from $Sub_{11}$, and one 
precoding vector pair from $Sub_{13}$. 
As to the remaining subsets, i.e., $Sub_{12}$ and $Sub_{14}$, we can only select one
more precoding vector pair, since the remaining receive signal dimensions at 
\emph{Alice} and \emph{Bob} are $N_a^r-3=1$ and $N_b^r-3=1$, respectively.
It is easy to verify that a sum S.D.o.F. of 7 can be achieved.

Concluding, an algorithm for constructing $({\bf V}_a, {\bf V}_b)$
which achieves the maximum sum S.D.o.F. is given in Table II.

%
%

\section{Numerical Results}

\begin{figure}[!t]
\centering
\includegraphics[width=3in]{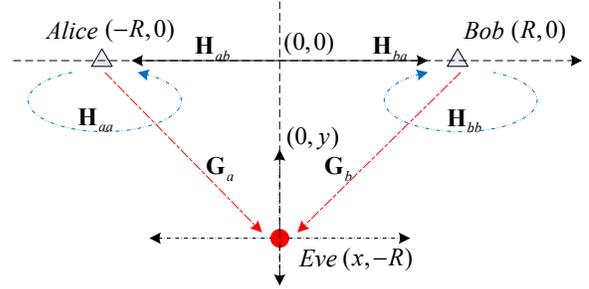}
\DeclareGraphicsExtensions. \caption{Model used in numerical experiments.}
\vspace* {-6pt}
\end{figure}

We consider a system model as illustrated in Fig. 2.
\emph{Alice} and \emph{Bob} are fixed at coordinates $(-R,0)$ and $(R,0)$ (unit: meters), $R=5$, respectively.
\emph{Eve} moves in two directions, i.e., along the $x$-coordinate from $(-15,-R)$ to $(15,-R)$, and
along the $y$-coordinate from $(0,2R)$ to $(0,0)$. Results are obtained over $10,000$ Monte Carlo runs as follows.
In each run, the channels are modeled as multipath flat fading. The effect of the
channel between any transmit-receive pair on the transmitted signal
is modeled by a multiplicative scalar of the form $d^{-c/2}e^{j\theta}$ \cite{Inaltekin09}, where
$d$ is the distance between the two nodes, $c$ is the path loss exponent and $\theta$ is a random phase,
which is taken to be uniformly distributed
within $[0, 2 \pi)$. The value of $c$ is typically in the range of 2 to 4. In our simulations we set $c=3.5$.
We assume that the distances of different combinations of transmit-receive antennas corresponding to the same
link are the same, and as such the corresponding path loss is the same.
The transmit power of each transmitting node is $P= 0$dBm. At each source, power is
equally allocated between different signal streams.
The noise power level is set as $\sigma^2=-60$dBm.
Unless otherwise specified, we set $N_a^t=3$, $N_b^r=2$, $N_b^t=3$, $N_a^r=2$, $N_e=5$.

In each figure to follow, we plot the average achievable sum secrecy transmission rate of the proposed scheme.
According to Table II, one can see that with our proposed cooperative transmission
scheme, a maximum sum S.D.o.F. of 2, i.e., an S.D.o.F. pair (1,1), can be achieved.
We compute the precoding matrix pair $({\bf V}_a, {\bf V}_b)$ by Table II,
with which we compute the achievable secrecy transmission rate of each user according to (\ref{eqSR}).
Exact knowledge of the channels is assumed in the computation.
For comparison, we also plot the average achievable sum secrecy transmission rate by some other schemes,
i.e., the one-way scheme, the wiretapped signal leakage minimization (WLSM) scheme by \cite{Tung15},
the match filter (MF) scheme and the zero-forcing (ZF) scheme.
In particular, in the one-way scheme, we set $N_a^t=5$, $N_b^t=1$ and $N_b^r=4$.
The one-way scheme can be regarded as a special case
of the proposed scheme where only \emph{Bob} operates in FD mode, and thus we rerun the simulations by Table II.
We should note that by Table II, the maximum achievable S.D.o.F.
of the one-way scheme for different number of transmit/receive antennas at \emph{Bob} and subject to $N_b=5$, is equal to 1.
The WLSM scheme \cite{Tung15} 
optimizes the transmit/receive filters iteratively,
for the purpose of minimizing the sum of the co-channel interference power
and the message signal power leaked to \emph{Eve}. Since the proposed scheme provides
closed-form precoding matrices, it has a computational advantage over the WLSM scheme.
For the MF scheme, we select the eigenvector corresponding the maximum eigenvalue of the legitimate channels,
i.e., ${\bf H}_{ab}$ and ${\bf H}_{ba}$, as the beamforming vector.
For the ZF scheme, we select the vector falling into the null space of the self-interfering channels, i.e.,
${\bf H}_{aa}$ and ${\bf H}_{bb}$, as the beamforming vector.

Figs. 3-5 illustrate the average achievable secrecy transmission rate versus the position of
\emph{Eve} along the $x$-coordinate, for different values of the self-interference parameter $\rho$.
It shows that the proposed transmission scheme outperforms
all the other schemes.
Interestingly, for the case of $x=0$, the proposed schemes achieve a local maximum secrecy transmission rate, while
the other schemes achieve a minimum secrecy transmission rate. This suggests that, the positions with $x=0$ are the most favorable ones.
On comparing Figs. 3-5, one can see that the average achievable
secrecy transmission rate of the proposed scheme degrades slightly with an increasing value of $\rho$, while
that of WLSM scheme degrades sharply. This suggests that as compared with the WLSM scheme, the proposed scheme is more robust
to self-interference.

\begin{figure}[!t]
\centering
\includegraphics[width=3in]{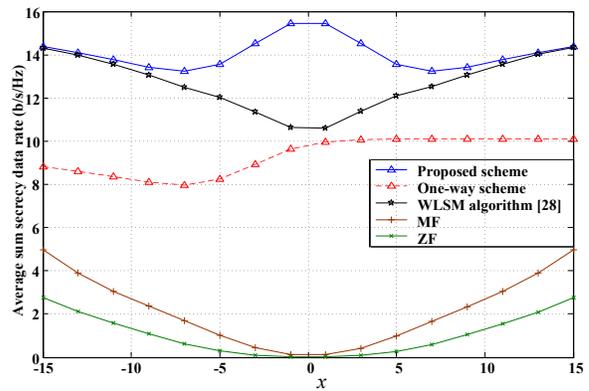}
\DeclareGraphicsExtensions. \caption{Average achievable secrecy rate versus the $x$-coordinate of
\emph{Eve}. The self-interference parameter $\rho=0$.} 
\vspace* {-6pt}
\end{figure}

\begin{figure}[!t]
\centering
\includegraphics[width=3in]{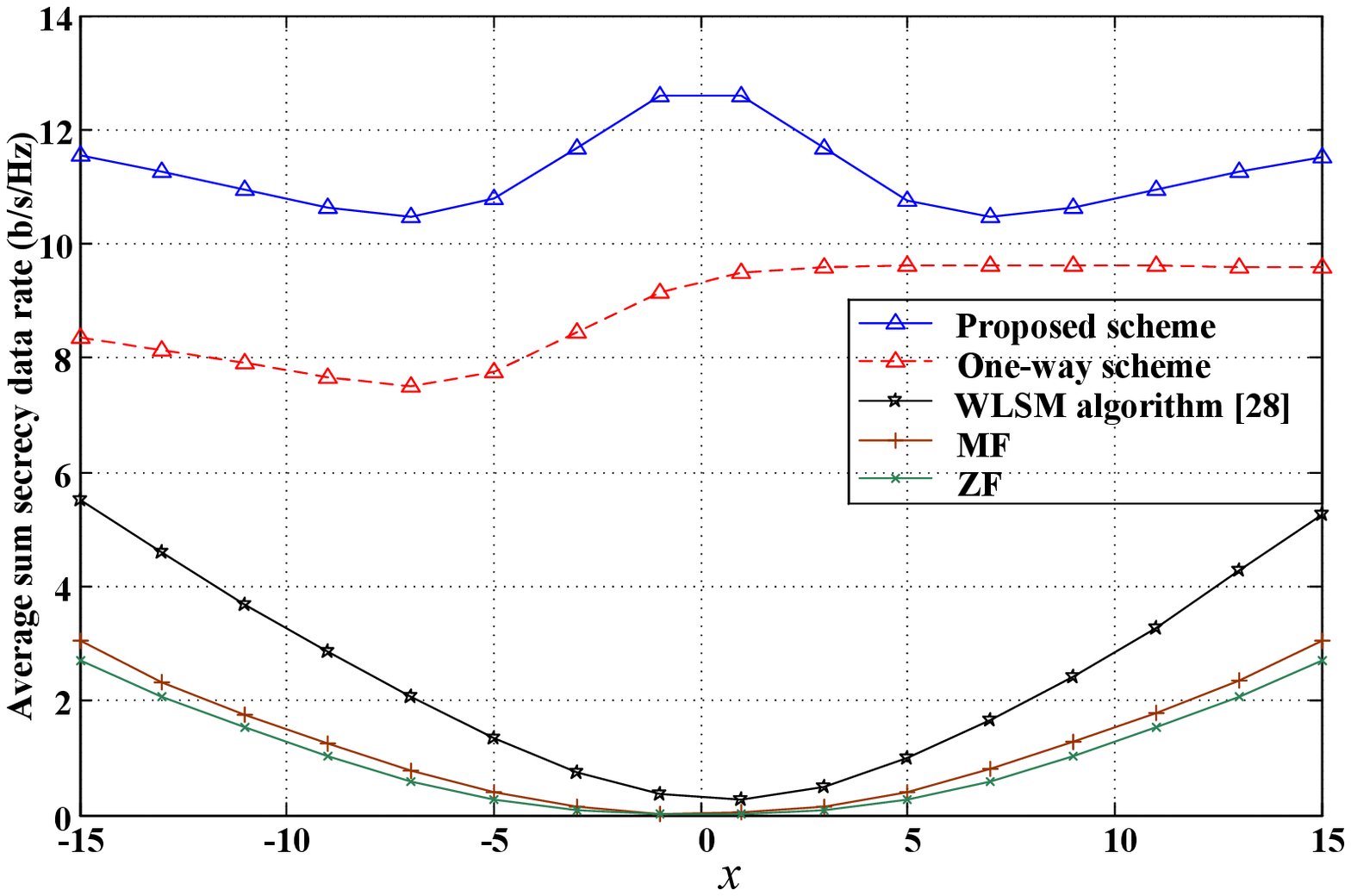}
\DeclareGraphicsExtensions. \caption{Average achievable secrecy rate versus the $x$-coordinate of
\emph{Eve}. The self-interference parameter is $\rho=0.1$.} 
\vspace* {-6pt}
\end{figure}

\begin{figure}[!t]
\centering
\includegraphics[width=3in]{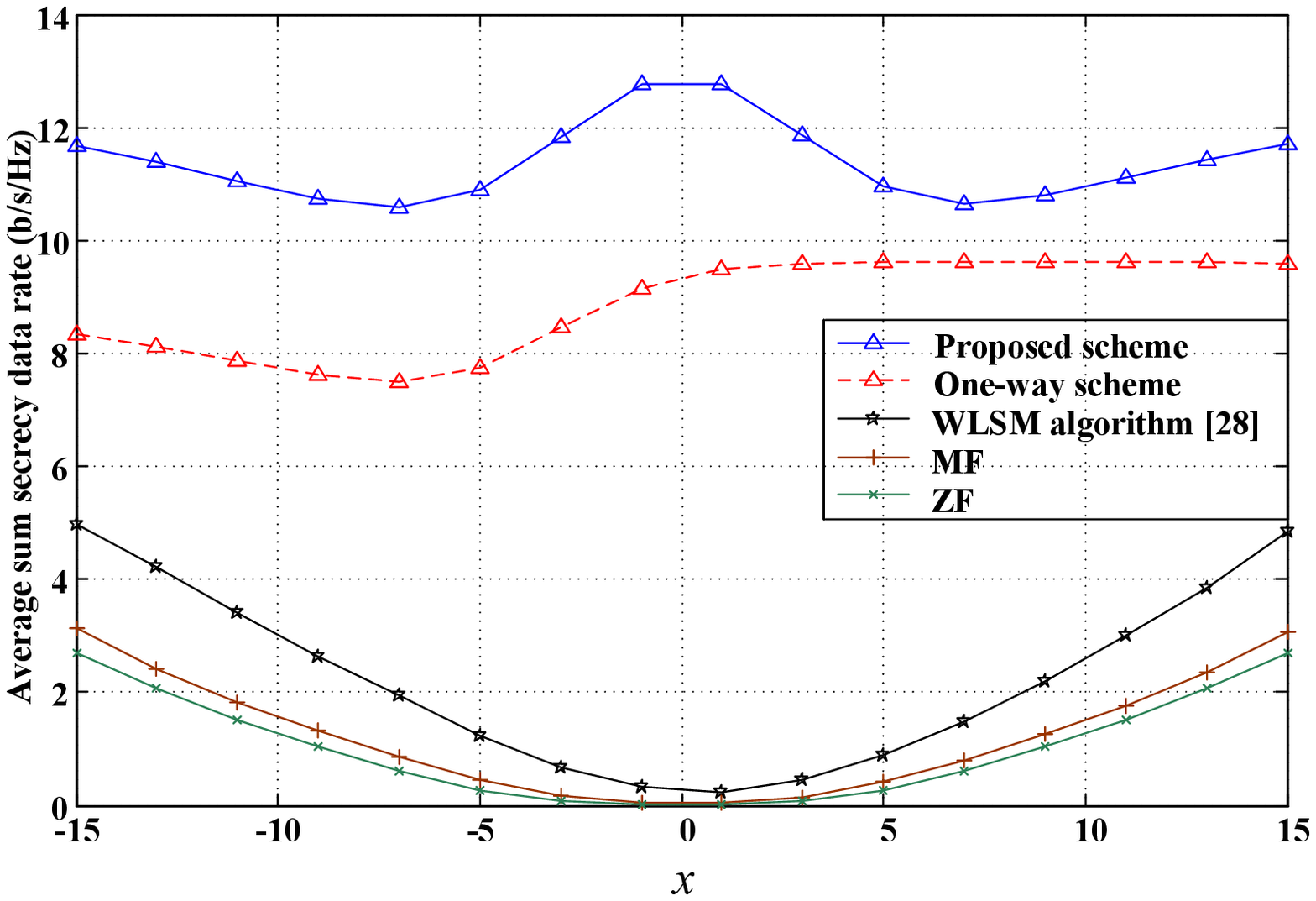}
\DeclareGraphicsExtensions. \caption{Average achievable secrecy rate versus the $x$-coordinate of
\emph{Eve}. The self-interference parameter is $\rho=1$.} 
\vspace* {-6pt}
\end{figure}

\begin{figure}[!t]
\centering
\includegraphics[width=3in]{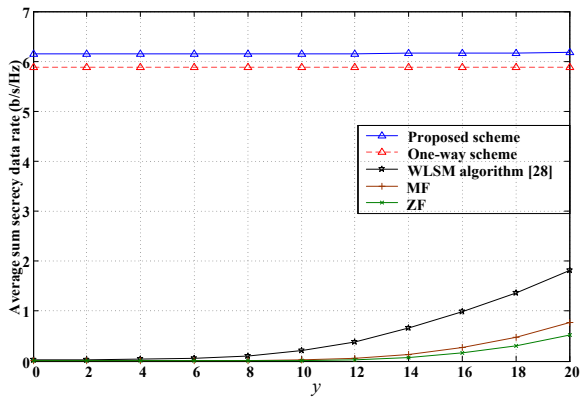}
\DeclareGraphicsExtensions. \caption{Average achievable secrecy rate versus the $y$-coordinate of
\emph{Eve}. The self-interference parameter is $\rho=1$.} 
\vspace* {-6pt}
\end{figure}


Fig. 6 illustrates the average achievable secrecy transmission rate versus the position of
\emph{Eve} along the $y$-coordinate for the strong self-interference case. 
We set $R=10$. 
It can be seen that, for both cases and with a decreasing value of $y$, the achievable secrecy transmission
rate of the proposed scheme remains constant. In contrast, the achievable secrecy transmission
rate of the other schemes dwindles sharply, and it is almost zero as $y$ approaches zero.
We have the following explanation. As \emph{Eve} moves along $y$-coordinate and closer to
both \emph{Alice} and \emph{Bob}, the message signal power received by \emph{Eve} improves.
The proposed scheme aligns the message signal and the
co-channel interference signal at \emph{Eve}, and thus keeps \emph{Eve}'s eavesdropping capability
constant. In contrast, the other schemes do not perform such signal alignment, and so their
achievable secrecy transmission rate decreases.

According to existing knowledge on wireless communications, in order
to tell apart all the signal steams, the sum number of signal streams which
the legitimate receiver receives should be no greater than the total number of receive antennas, i.e,
\begin{subequations}
\begin{align}
&{\rm {rank}} \{{\bf H}_{aa}{\bf V}_a\}+{\rm {rank}} \{{\bf H}_{ab}{\bf V}_b\} \le N_a^r, \label{eq37a} \\
&{\rm {rank}} \{{\bf H}_{ba}{\bf V}_a\}+{\rm {rank}} \{{\bf H}_{bb}{\bf V}_b\} \le N_b^r. \label{eq37b}
\end{align}
\end{subequations}
Indeed, a sufficient condition for the proposed scheme to stop selecting
is that both (\ref{eq37a}) and (\ref{eq37b}) are violated.
In Fig. 7, we compare the average achievable secrecy rate
of two different schemes, i.e., the proposed scheme, and the proposed scheme subject to (\ref{eq37a})(\ref{eq37b}).
Here, we set $N_a^t=4$, $N_b^r=2$, $N_b^t=5$, $N_a^r=3$, $N_e=5$.
By Table I, it holds that $d_{22}=2$, $d_{23}=d_{24}=1$, and all the other $d_{ij}$'s are zero.
According to Table II, for the proposed scheme, we will
respectively select one precoding vector pair from $Sub_{22}$ and $Sub_{23}$, and
an S.D.o.F. pair (1, 2) can be achieved; for the case subject to (\ref{eq37a})(\ref{eq37b}), we will select one precoding vector pair from
$Sub_{22}$ or $Sub_{23}$, and an S.D.o.F. pair (1, 1) can be achieved.
We rerun the simulations, with the $x$-coordinate of
\emph{Eve} varying from $(-40,-R)$ to $(40,-R)$.
Fig. 7 shows that, except for the case in which \emph{Eve} is in a medium distance from \emph{Alice} or \emph{Bob},
the proposed scheme outperforms that with constraints.
This can be explained as follows. First, when \emph{Eve} is close to
\emph{Alice} and \emph{Bob}, the co-channel interference is strong, and it helps shield the message signal from \emph{Eve}.
Thus, the proposed scheme, which achieves a greater
sum S.D.o.F., outperforms that the scheme with the constraints in (\ref{eq37a})(\ref{eq37b}).
Second, when \emph{Eve} moves to the left and in a medium distance from \emph{Bob}, at \emph{Eve} the
co-channel interference power by \emph{Bob} becomes smaller, and considering multiple signal streams
at \emph{Bob} would worsen this situation, which results in worse shielding of
the message signal from \emph{Alice}. Thus, the scheme with the constraints in (\ref{eq37a})(\ref{eq37b})
outperforms the proposed scheme.
Third, when \emph{Eve} is far enough away from \emph{Alice} and \emph{Bob},
it almost receives nothing. Thus, the message signals from \emph{Alice} and \emph{Bob} are naturally secure.
These observations give us another clue showing that, in contrast to a harm role in the network without secrecy constraints,
CCI acts positively in the network with secrecy constraints.
\begin{figure}[!t]
\centering
\includegraphics[width=3in]{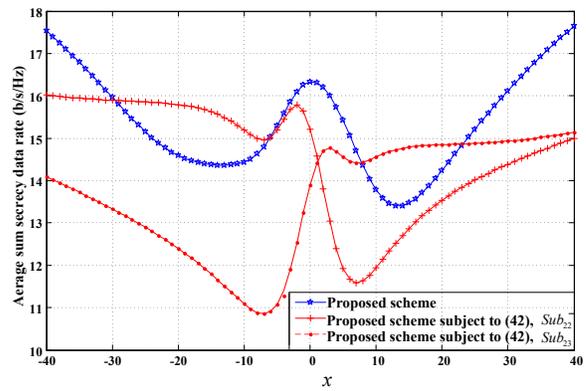}
\DeclareGraphicsExtensions. \caption{Average achievable secrecy rate versus the $x$-coordinate of
\emph{Eve}. The self-interference parameter is $\rho=1$.} 
\vspace* {-6pt}
\end{figure}

\begin{figure}[!t]
\centering
\includegraphics[width=3in]{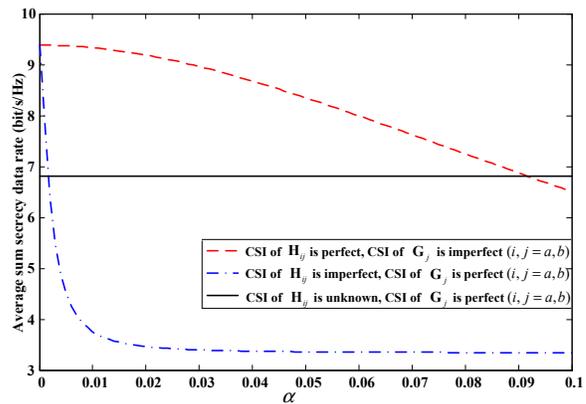}
\DeclareGraphicsExtensions. \caption{Average achievable secrecy rate versus the channel uncertainty.
The self-interference parameter is $\rho=1$.} 
\vspace* {-6pt}
\end{figure}

In Fig. 8, we examine the secrecy rate performance in the presence of imperfect channel estimates.
{We model imperfect CSI through a Gauss-Markov uncertainty of the form \cite{Nosrat11}
\begin{align}
{\bf G}_{i}=d_{i}^{-c/2}\left(\sqrt{1-\alpha^2}\bar{\bf G}_{i}+\alpha\Delta\bar{\bf G}_{i}\right), i=a, b, \label{eqImpCSI}
\end{align}
where $0 \le \alpha \le 1$ denotes the channel uncertainty. $\alpha=0$ and $\alpha=1$
correspond to perfect channel knowledge and no CSI knowledge, respectively.}
The entries of $\bar{\bf G}_{i}$ are $e^{j\theta}$ with $\theta$ be a random phase
uniformly distributed within $[0, 2 \pi)$.
$\Delta\bar{\bf G}_{i}\sim \mathcal{CN}(\bf{0},\bf{I})$
represents the Gaussian error channel matrices.
$d_{i}$ denotes the distance from \emph{Alice} or \emph{Bob}.
With the same channel model as in (\ref{eqImpCSI}), we model the channel uncertainty of the channels ${\bf H}_{ij}$, $i, j=a, b$.
We reset $N_a^t=N_b^t=4$, $N_a^r=N_b^r=3$, and $N_e=4$. According to Table II, an S.D.o.F. pair of (2, 2) can be achieved.
We construct the precoding matrices ${\bf V}_a$ and ${\bf V}_b$ with the estimated channels.
It can be observed that the achievable secrecy rate drops with the
increase of uncertainty in the channels ${\bf G}_{i}$, $i=a,b$, or the channels ${\bf H}_{ij}$, $i, j=a, b$.
This should be expected, since the eavesdropping channels ${\bf G}_{i}$, $i=a,b$,
and also the self-interference channels ${\bf H}_{ii}$, $i=a,b$, enter in the construction of the precoding matrices.
We should note that, when the self-interfering channels, i.e., ${\bf H}_{ii}$, $i=a,b$, are unknown, by some slight changes
the proposed scheme still works. In particular, since ${\bf H}_{ii}$, $i=a,b$ are unknown, we are not able to
obtain a precoding vector along which the message signal
does not interfere with the unintended user, and so in Table I, $d_{11}=d_{13}=d_{21}=d_{22}=d_{23}=0$.
Substituting the other $d_{ij}$'s into Table II, we can construct a precoding matrix pair which is independent
of the channels ${\bf H}_{ij}$, $i, j=a, b$. As expected, in Fig. 8 it shows that the achievable
secrecy rate remains unchanged. 

\section{Conclusion}
We have examined the maximum achievable sum secrecy degrees
of freedoms (sum S.D.o.F.) for a multiple-input multiple-output (MIMO) Gaussian wiretap channel,
where \emph{Alice} and \emph{Bob} operate in FD mode, i.e., exchanging
confidential messages at the same time, and a passive eavesdropper who wants to wiretap the confidential messages
from both \emph{Alice} and \emph{Bob}.
We have addressed analytically the sum S.D.o.F. maximization problem.
We also have constructed precoding matrix pairs which achieve the maximum sum S.D.o.F..
Numerical results have revealed the advantages of the proposed secrecy transmission
scheme over existing schemes. The proposed scheme outperforms
all comparison schemes in terms of the achievable average secrecy transmission rate. Since the proposed
scheme provides closed-form precoding matrix pairs, it also has a computational advantage over the WLSM
scheme proposed by \cite{Tung15}. Also, the proposed secrecy transmission scheme is robust to
self-interference, and also robust to the conventional vulnerable positions of \emph{Eve}, i.e., the position
with $x=0$. Further, if properly designed, co-channel interference is helpful in improving the overall
secrecy rate throughput. Finally, apart from the advantage of higher spectral efficiency,
the FD based network also provides a good structure in terms of keeping messages secret.

\appendices

\section{Mathematical Background on GSVD}
Given two full rank matrices ${\bf A}\in {{\mathbb C} ^{N \times M}}$ and
${\bf B}\in {{\mathbb C} ^{N \times K}}$. It holds that
\begin{subequations}
\begin{align}
k\triangleq &  {\rm {rank}}\{[({\bf A}^H)^T, ({\bf B}^H)^T]^T\}=\min\{M+K,N\}, \label{eq001a}\\
p\triangleq & {\rm {dim}}\{ {\rm {span}}({\bf A})^\perp  \cap {\rm {span}}({\bf B}) \}=k- \min \{M,N\}, \label{eq001b} \\
r \triangleq &  {\rm{dim}}\{{\rm {span}}({\bf A})\cap {\rm {span}}({\bf B})^\perp\}=k- \min \{K,N\}, \label{eq001c}\\
s \triangleq & {\rm {dim}} \{{\rm {span}}({\bf A})\cap {\rm {span}}({\bf B})\}=k-p-r \nonumber \\
 &=(\min \{M,N\}+\min \{K,N\}-N)^+. \label{eq001d}  %
\end{align}
\end{subequations}
The GSVD of $({\bf A}^H ,{\bf B}^H )$ \cite{Paige81} returns
unitary matrices ${\bf\Psi}_1 \in {{\mathbb C} ^{M\times M}}$ and ${\bf\Psi}_2 \in {{\mathbb C} ^{K\times K}}$,
positive diagonal matrices ${\bf \Lambda}_1\in \mathbb{R}^{s\times s}$
and ${\bf \Lambda}_2\in \mathbb{R}^{s\times s}$, with ${\bf{\Lambda}}_1^H{{\bf{\Lambda}}_1} + {\bf{\Lambda}}_2^H{{\bf{\Lambda}}_2} = {\bf{I}}$,
and a matrix ${\bf X} \in {{\mathbb C} ^{N\times k}}$ with ${\rm{rank}}\{{\bf X}\}=k$, such that
\begin{subequations}
\begin{align}
& {\bf A}[{\bf \Psi}_{11}\ {\bf \Psi}_{12}\ {\bf \Psi}_{13}]=[{\bf X}_{1}{\bf 0}\ {\bf X}_{2}{\bf \Lambda}_1^H\ {\bf X}_{3}{\bf I}_r], \label{eq13a}\\
& {\bf B}[{\bf \Psi}_{21}\ {\bf \Psi}_{22}\ {\bf \Psi}_{23}]=[{\bf X}_{1}{\bf I}_p\ {\bf X}_{2}{\bf \Lambda}_2^H\ {\bf X}_{3}{\bf 0}]. \label{eq13b}
\end{align}
\end{subequations}
Here ${\bf \Psi}_{11}$, ${\bf \Psi}_{12}$ and ${\bf \Psi}_{13}$ are
the first $M-s-r$, the following $s$, and the remaining $r$ columns of of ${\bf \Psi}_1$, respectively;
${\bf \Psi}_{21}$, ${\bf \Psi}_{22}$ and ${\bf \Psi}_{23}$ are the
first $p$, the following $s$, and the remaining $K-s-p$ columns of ${\bf \Psi}_2$, respectively.
In addition, ${\bf X}_{1}$, ${\bf X}_{2}$ and ${\bf X}_{3}$ denote the
first $p$, the following $s$, and the remaining $r$ columns of ${\bf X}$, respectively.

With the GSVD decomposition, one can decompose the union
of ${\rm {span}}({\bf A})$ and ${\rm {span}}({\bf B})$ into three
subspaces, as shown in Fig. 9.

\begin{figure}[!t]
\centering
\includegraphics[width=3in]{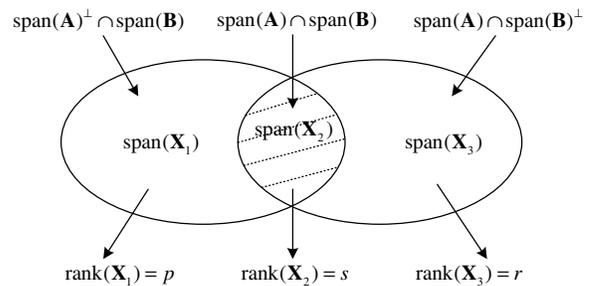}
\DeclareGraphicsExtensions. \caption{The geometric relationship between the subspaces ${\rm{span}}({\bf A})$ and ${\rm{span}}({\bf B})$.}
\vspace* {-12pt}
\end{figure}

\section{Proof of \emph{Proposition} 1} \label{appA}

By definition, $d_s^{\rm sum}\ge \bar{d}_s^{\rm sum}$.
In what follows, we will show that for any given precoding matrix pair $({\bf V}_a,{\bf V}_b)\in {\mathcal I}$,
one can always find another precoding matrix pair $({\bf V}_a^\prime,{\bf V}_b^\prime) \in \bar {\mathcal I}$,
such that $d_s^a({\bf V}_a,{\bf V}_b) \le d_s^a({\bf V}_a^\prime,{\bf V}_b^\prime)$ and
$d_s^b({\bf V}_a,{\bf V}_b)\le d_s^b({\bf V}_a^\prime,{\bf V}_b^\prime)$,
which indicates that $d_s^{\rm sum}\le \bar{d}^{\rm sum}$.
In this way, we prove that $d_s^{\rm sum}=\bar{d}^{\rm sum}$.

The basic idea for constructing the precoding matrix pair $({\bf V}_a^\prime,{\bf V}_b^\prime) \in \bar {\mathcal I}$
is to exclude the subspaces ${\rm{span}}({\bf G}_a{{\bf V}_a})\setminus{\rm{span}}({\bf G}_b{\bf V}_b)$ and
${\rm{span}}({\bf G}_b{\bf V}_b)\setminus{\rm{span}}({\bf G}_a{{\bf V}_a})$, without decreasing the
S.D.o.F. pair.

To that objective, firstly, by letting ${\bf A}={\bf G}_b{\bf V}_b$ and ${\bf B}={\bf G}_a{\bf V}_a$, and applying
the GSVD decomposition in Appendix A, we arrive at
\begin{subequations}
\begin{align}
&d_s^a({\bf V}_a,{\bf V}_b)= m_1({\bf V}_a,{\bf V}_b)-n_1({\bf V}_a,{\bf V}_b) \nonumber\\
&= m_1({\bf V}_a,{\bf V}_b)
-{\rm {rank}}\{\hat{\bf\Psi}_{21}\} \label{eqa4b}\\
&\le {\rm{dim}}\{{\rm{span}}({\bf{H}}_{ba}{{\bf V}_a}{\hat {\underline{\bf\Psi}}}_{21})\setminus{\rm{span}}({\bf{H}}_{bb}{\bf V}_b)\} \label{eqa4c}\textrm{,}
\end{align}
\end{subequations}
where ${\hat {\underline{\bf\Psi}}}_{21} \triangleq [\hat{\bf\Psi}_{22},\hat{\bf\Psi}_{23}]$.
Here, $\hat{\bf\Psi}_{21}$, $\hat{\bf\Psi}_{22}$ and $\hat{\bf\Psi}_{23}$ correspond to
${\bf\Psi}_{21}$, ${\bf\Psi}_{22}$ and ${\bf\Psi}_{23}$, and arise due to the GSVD of
$({\bf G}_b{\bf V}_b)^H$ and $({\bf G}_a{\bf V}_a)^H$.
(\ref{eqa4c}) holds true, because $m_1({{\bf V}_a},{\bf V}_b)
\le m_1({{\bf V}_a}{\hat {\underline{\bf\Psi}}}_{21},{\bf V}_b)
+m_1({{\bf V}_a}\hat{\bf\Psi}_{21},{\bf V}_b)$ and $m_1({{\bf V}_a}\hat{\bf\Psi}_{21},{\bf V}_b)
\le {\rm {rank}}\{\hat{\bf\Psi}_{21}\}$.

Secondly, since $n_1({\bf V}_a,{\bf V}_b)={\rm{span}}({\bf G}_a{{\bf V}_a}\hat{\bf\Psi}_{21})$, it holds that
$n_2({\bf V}_a,{\bf V}_b)=n_2({\bf V}_a\hat{\underline{\bf\Psi}}_{21},{\bf V}_b)$. Thus,
\begin{subequations}
\begin{align}
&d_s^b({\bf V}_a,{\bf V}_b)= m_2({\bf V}_a,{\bf V}_b)
-n_2({\bf V}_a\hat{\underline{\bf\Psi}}_{21},{\bf V}_b) \label{eqa5a}\\
&\le m_2({\bf V}_a\hat{\underline{\bf\Psi}}_{21},{\bf V}_b)
-n_2({\bf V}_a\hat{\underline{\bf\Psi}}_{21},{\bf V}_b) \label{eqa5b}\\
&=m_2({\bf V}_a\hat{\underline{\bf\Psi}}_{21},{\bf V}_b)
-{\rm {rank}}\{\hat{\bf\Psi}_{13}\} \label{eqa5c}\\
&\le {\rm{dim}}\{{\rm{span}}({\bf{H}}_{ab}{\bf V}_b
\hat{\underline{\bf\Psi}}_{13})\setminus{\rm{span}}({\bf{H}}_{aa}{{\bf V}_a}{\hat {\underline{\bf\Psi}}}_{21})\} \label{eqa5d}\textrm{,}
\end{align}
\end{subequations}
where ${\hat {\underline{\bf\Psi}}}_{13} \triangleq [\hat{\bf\Psi}_{11},\hat{\bf\Psi}_{12}]$.
Here, $\hat{\bf\Psi}_{11}$, $\hat{\bf\Psi}_{12}$ and $\hat{\bf\Psi}_{13}$ correspond to
${\bf\Psi}_{11}$, ${\bf\Psi}_{12}$ and ${\bf\Psi}_{13}$, and together with (\ref{eqa5c}) arise due to the GSVD of
$({\bf G}_b{\bf V}_b)^H$ and $({\bf G}_a{\bf V}_a)^H$.
Since $m_2({{\bf V}_a}{\hat {\underline{\bf\Psi}}}_{21},{\bf V}_b\hat{{\bf\Psi}}_{13})\le {\rm {rank}}\{\hat{\bf\Psi}_{13}\} $
and $m_2({\bf V}_a\hat{\underline{\bf\Psi}}_{21},{\bf V}_b) \le
 m_2({{\bf V}_a}{\hat {\underline{\bf\Psi}}}_{21},{\bf V}_b\hat{\underline{\bf\Psi}}_{13})
+ m_2({{\bf V}_a}{\hat {\underline{\bf\Psi}}}_{21},{\bf V}_b\hat{{\bf\Psi}}_{13})$, one can see that
(\ref{eqa5d}) holds true.

Let ${\bf V}_a^\prime={\bf V}_a{\hat {\underline{\bf\Psi}}}_{21}$,
${\bf V}_b^\prime={\bf V}_b\hat{\underline{\bf\Psi}}_{13}$. By definition, we have
$({\bf V}_a^\prime, {\bf V}_b^\prime) \in  \bar{\mathcal{I}}$, and
\begin{align}
{\rm{span}}({\bf G}_a{\bf V}_a{\hat {\underline{\bf\Psi}}}_{21})
= {\rm{span}}({\bf G}_b{\bf V}_b\hat{\underline{\bf\Psi}}_{13}).\label{eqa6}
\end{align}
Therefore,
\begin{align}
&d_s^a({\bf V}_a^\prime, {\bf V}_b^\prime)=
{\rm{dim}}\{{\rm{span}}({\bf{H}}_{ba}{{\bf V}_a}{\hat {\underline{\bf\Psi}}}_{21})
\setminus{\rm{span}}({\bf{H}}_{bb}{\bf V}_b\hat{\underline{\bf\Psi}}_{13})\}, \nonumber \\
&d_s^b({\bf V}_a^\prime, {\bf V}_b^\prime)=
{\rm{dim}}\{{\rm{span}}({\bf{H}}_{ab}{\bf V}_b\hat{\underline{\bf\Psi}}_{13})
\setminus{\rm{span}}({\bf{H}}_{aa}{{\bf V}_a}{\hat {\underline{\bf\Psi}}}_{21})\}, \nonumber
\end{align}
which together with (\ref{eqa4c}) and (\ref{eqa5d}), indicate that $d_s^a({\bf V}_a,{\bf V}_b) \le d_s^a({\bf V}_a^\prime, {\bf V}_b^\prime)$
and $d_s^b({\bf V}_a,{\bf V}_b) \le d_s^b({\bf V}_a^\prime, {\bf V}_b^\prime)$, respectively.
This completes the proof.

\section{Proof of \emph{Corollary} 1} \label{appB}
In what follows, we will show that for any given matrix pair $({\bf V}_a,{\bf V}_b )\in \bar{\mathcal I}$, one can
always construct another precoding matrix pair $({\bf V}_a^\star,{\bf V}_b^\star)\in \hat{\mathcal I}$ where
${\bf G}_a{\bf V}_a^\star={\bf G}_b{\bf V}_b^\star$,
such that $d_s^a({\bf V}_a^\star,{\bf V}_b^\star)= d_s^a({\bf V}_a ,{\bf V}_b )$
and $d_s^b({\bf V}_a^\star,{\bf V}_b^\star)= d_s^b({\bf V}_a ,{\bf V}_b )$, which indicates that
$\bar{d}_s^{\rm sum}= \hat{d}_s^{\rm sum}$.

For any given $({\bf V}_a ,{\bf V}_b)\in \bar{\mathcal I}$, ${\bf V}_a \in {\mathbb {C}}^{N_a^t \times K_a}$,
${\bf V}_b \in {\mathbb {C}}^{N_b^t \times K_b}$, we should have
\begin{align}
&{\rm{span}}({\bf G}_a{{\bf V}_a} ) = {\rm{span}}({\bf G}_b{\bf V}_b ). \label{eqb1}
\end{align}
Since all channel matrices are assumed to be full rank,
it holds that ${\rm {rank}}\{{\bf G}_b{\bf V}_b\}=\min\{K_b,N_e\}$.
In the sequel, we will consider two distinct cases, i.e., $K_b\ge N_e$ and $K_b< N_e$.

\subsection{For the case of $K_b\ge N_e$}
It holds that ${\rm {rank}}\{{\bf G}_b{\bf V}_b\}=N_e$.
Denote
\begin{align}
&{\bf G}_a{\bf V}_a=\left[{\bf U}_{a1}\ {\bf U}_{a0} \right]
\left[ {\begin{array}{*{20}{c}}
{{{\bf{\Sigma}}_{a1}}}&{\bf{0}}\\
{\bf{0}}&{{{\bf{0}}}}
\end{array}} \right]
\left[ {\begin{array}{*{20}{c}}
{{{\bf{T}}_{a1}^H}}\\
{{{\bf{T}}_{a0}^H}}
\end{array}} \right], \nonumber \\
&{\bf G}_b{\bf V}_b=\left[{\bf U}_{b1}\ {\bf U}_{b0} \right]
\left[ {\begin{array}{*{20}{c}}
{{{\bf{\Sigma}}_{b1}}}&{\bf{0}}\\
{\bf{0}}&{{{\bf{0}}}}
\end{array}} \right]
\left[ {\begin{array}{*{20}{c}}
{{{\bf{T}}_{b1}^H}}\\
{{{\bf{T}}_{b0}^H}}
\end{array}} \right], \nonumber
\end{align}
as the SVD of ${\bf G}_a{\bf V}_a$ and ${\bf G}_b{\bf V}_b$, respectively.
Then, the matrices ${\bf G}_a{\bf V}_a{\bf{T}}_{a1}$ and ${\bf G}_b{\bf V}_b{\bf{T}}_{b1}$ are invertible.

Due to (\ref{eqb1}), it holds that ${\rm{span}}({\bf G}_a{{\bf V}_a} {\bf{T}}_{a1}) = {\rm{span}}({\bf G}_b{\bf V}_b {\bf{T}}_{b1})$.
Thus, there exists some invertible matrix $\bf A$, such that ${\bf G}_a{{\bf V}_a}{\bf{T}}_{a1}{\bf A} ={\bf G}_b{\bf V}_b{\bf{T}}_{b1} $.
\begin{enumerate}
\item If $K_b \ge K_a$, let ${\bf V}_a^\star={\bf V}_a[{\bf{T}}_{a1}{\bf A} \ {\bf{T}}_{a0} \ {\bf 0}_{K_a \times (K_b-K_a)}]$
and ${\bf V}_b^\star={\bf V}_b[{\bf{T}}_{b1} \ {\bf{T}}_{b0}]$.
\item If $K_b < K_a$, let ${\bf V}_b^\star={\bf V}_b[{\bf{T}}_{b1} \ {\bf{T}}_{b0}\ {\bf 0}_{K_b \times (K_a-K_b)}]$
and ${\bf V}_a^\star={\bf V}_a[{\bf{T}}_{a1}{\bf A} \ {\bf{T}}_{a0}]$.
\end{enumerate}
It can be verified that ${\bf G}_a{{\bf V}_a}^\star ={\bf G}_b{\bf V}_b^\star$ holds true for both cases.
Moreover, since both $[{\bf{T}}_{a1}{\bf A} \ {\bf{T}}_{a0}]$ and $[{\bf{T}}_{b1} \ {\bf{T}}_{b0}]$ are
invertible matrices, it holds that $d_s^a({\bf V}_a^\star,{\bf V}_b^\star)= d_s^a({\bf V}_a ,{\bf V}_b )$
and $d_s^b({\bf V}_a^\star,{\bf V}_b^\star)= d_s^b({\bf V}_a ,{\bf V}_b )$.

\subsection{For the case of $K_b < N_e$}
It holds that ${\bf G}_a{\bf V}_a$ and ${\bf G}_b{\bf V}_b$ are full column rank.
Let ${\bf P}_v$ and ${\bf P}_w$ be the projection matrix of ${\bf G}_a{\bf V}_a$ and ${\bf G}_b{\bf V}_b$, respectively, i.e.,
\begin{subequations}
\begin{align}
&{\bf P}_v={\bf G}_a{{\bf V}_a} (({\bf G}_a{{\bf V}_a} )^H{\bf G}_a{{\bf V}_a} )^{-1}({\bf G}_a{{\bf V}_a} )^H,\nonumber \\
&{\bf P}_w={\bf G}_b{\bf V}_b (({\bf G}_b{\bf V}_b )^H{\bf G}_b{\bf V}_b )^{-1}({\bf G}_b{\bf V}_b )^H. \nonumber
\label{eqb2b}
\end{align}
\end{subequations}
Let ${\bf V}_a^\star={\bf V}_a{\bf B}$, with ${\bf B} = (({\bf G}_a{{\bf V}_a} )^H{\bf G}_a{{\bf V}_a} )^{-1}({\bf G}_a{{\bf V}_a} )^H$.
Let ${\bf V}_b^\star={\bf V}_b{\bf C}$, with ${\bf C}= (({\bf G}_b{\bf V}_b )^H{\bf G}_b{\bf V}_b )^{-1}({\bf G}_b{\bf V}_b )^H$.

By (\ref{eqb1}), it holds that ${\bf P}_v={\bf P}_w$. Thus, ${\bf G}_a{{\bf V}_a}{\bf B}={\bf G}_b{\bf V}_b{\bf C}$. Moreover,
since both ${\bf B}$ and ${\bf C}$ are full row rank, it holds that $d_s^a({\bf V}_a^\star,{\bf V}_b^\star)= d_s^a({\bf V}_a ,{\bf V}_b )$
and $d_s^b({\bf V}_a^\star,{\bf V}_b^\star)= d_s^b({\bf V}_a ,{\bf V}_b )$.

This completes the proof.

\bibliography{mybib}
\bibliographystyle{IEEEtran}

\end{document}